\newif\ifdcol
\documentclass[letterpaper,11pt]{article} \dcolfalse
 
\usepackage{ltexpprt} 
\usepackage{amssymb,amsmath,amsfonts}
\usepackage{epsfig}
\usepackage{bbold}
\usepackage{cite}
\usepackage{hyperref}
\usepackage{stmaryrd}
\usepackage[margin=1in,top=1in,bottom=1in]{geometry}
\usepackage{graphicx}

\newcommand{\bitm}{\begin{itemize}}
\newcommand{\eitm}{\end{itemize}}
\newcommand{\be}{\begin{equation}}
\newcommand{\ee}{\end{equation}}
\newcommand{\bea}{\begin{eqnarray}}
\newcommand{\eea}{\end{eqnarray}}
\newcommand\ba{\renewcommand{\arraystretch}{.7} \left[ \begin{array}{*{6}{@{\hspace{2pt}}r@{\hspace{1.5pt}}}}}
\def\ea{\end{array}\right]}
\def\nn{\nonumber\\}
\newcommand{\bfi}{\begin{figure}}
\newcommand{\efi}{\end{figure}}

\newcommand{\mat}[1]{\begin{bf} #1 \end{bf}}
\newcommand{\fn}{}

\DeclareMathOperator{\boxdim}{{boxdim}}

\DeclareMathOperator{\lboxdim}{\underline{boxdim}}

\DeclareMathOperator{\vol}{vol}
\allowdisplaybreaks

\def\qed  {{
\parfillskip=0pt 
\widowpenalty=10000 
\displaywidowpenalty=10000   
\finalhyphendemerits=0   
\leavevmode 
\unskip 
\nobreak 
\hfil 
\penalty50 
\hskip.2em 
\null 
\hfill 
$\square$
\par}} 

\usepackage{fancyheadings}
  \renewcommand{\headrulewidth}{0pt}
\begin{document}

\ifdcol
\else
\setcounter{page}{0}
\fi

\setlength{\abovedisplayskip}{4pt}
\setlength{\belowdisplayskip}{4pt}

\title{\Large Two Embedding Theorems for Data with Equivalences under Finite Group Action}
\author{Fabian Lim\thanks{F. Lim recieved support from NSF Grant ECCS-1128226.} 
\\Research Laboratory of Electronics, MIT, 
Cambridge, MA 02139, USA\\
flim@mit.edu
}
\date{}

\maketitle

\begin{abstract} 
There is recent interest in compressing data sets for non-sequential settings, where lack of obvious orderings on their data space, require notions of data equivalences to be considered. 
For example, Varshney \& Goyal (DCC, 2006) considered multiset equivalences, while Choi \& Szpankowski (IEEE Trans. IT, 2012) considered isomorphic equivalences in graphs.  
Here equivalences are considered under a relatively broad framework -
finite-dimensional, non-sequential data spaces with equivalences under group action,
for which analogues of two well-studied embedding theorems are derived: the Whitney embedding theorem and the Johnson-Lindenstrauss lemma.
Only the canonical data points need to be carefully embedded, each such point representing a set of data points equivalent under group action.
Two-step embeddings are considered.
First, a group invariant is applied to account for equivalences, and then secondly, a linear embedding takes it down to low-dimensions.
Our results require hypotheses on discriminability of the applied invariant, such notions related to
\emph{seperating invariants} (Dufresne, 2008), and \emph{completeness} in pattern recognition (Kakarala, 1992).

Our first theorem shows that almost all such two-step embeddings can one-to-one embed the canonical part of a bounded, discriminable set 
of data points, if embedding dimension exceeds $2 k$
whereby $k$ is the box-counting dimension of 
the set closure of canonical data points.
Our second theorem shows for $k$ equal to the number of canonical points of a finite data set, 
a randomly sampled two-step embedding, 
preserves isometries (of the canonical part) up to factors $1 \pm \epsilon$ with probability at least $1 - \beta$, if the embedding dimension exceeds $(2\log k + \log(1/\beta))/\alpha(\epsilon,\delta)$ for some function $\alpha$, and $\delta$ is a positive constant capturing a certain discriminability property of the invariant.
In the second theorem, the value $k$ is tied only to the canonical part, which may be significantly smaller than the ambient data dimension, up to a factor equal to the size the group.

\end{abstract}

\ifdcol
\else
\newpage
\fi
 
\renewcommand{\a}{\mat{a}}
\newcommand{\Norm}[2]{|| #1 ||_{{#2}}}
\newcommand{\Real}{\mathbb{R}}
\newcommand{\xbk}{\overline{\mat{a}}_{k}}

\newcommand{\Sens}{\Phi}
\renewcommand{\S}{\mathcal{S}}
\newcommand{\matt}[1]{\mat{#1}}
\newcommand{\x}{a}

\newcommand{\col}{\pmb{\phi}}
\newcommand{\y}{b}
\newcommand{\eps}{\epsilon}
\newcommand{\RIC}{\delta}
\renewcommand{\L}{\Real}

\section{Introduction}

A discrete finite sequence is 
arguably the most \emph{generic} mathematical representation for finite-dimensional data.
However, of recent interest are data sets where it is 
unclear how to appropriately assign sequence orderings to the data space.
For example, ranking data lives on a space of index subsets, which has no meaningful ordering~\cite{Diaconis}.
Graphical data lives on a space of graph edges, 
and node labellings may be often irrelevant~\cite{Kondor2008,Choi}.
\emph{Quotient spaces} that describe matrix manifolds, \textit{e.g.}, the \emph{Grassman manifold}, 
have equivalence classes as elements~\cite{Absil}. 

\begin{figure}
	\centering
 		\includegraphics[width=1\linewidth]{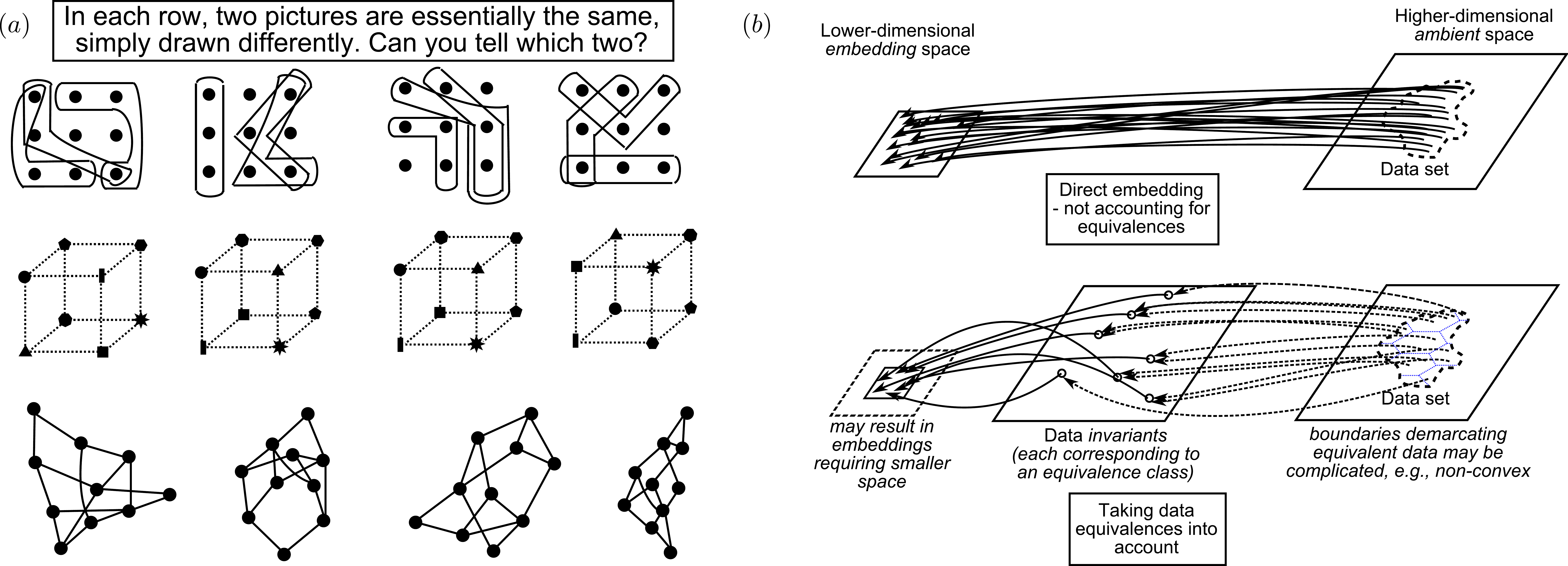}
	\caption{In $(a)$, an exercise illustrating equivalences between three types of ``non-conventional'' data (for answers see below). In $(b)$, accounting for data equivalences while performing embeddings.}
	\label{fig:Examples}
\end{figure}

{
\let\thefootnote\relax\footnotetext{\rotatebox{180}{First row : two-three. Second row : one-four. Third row : one-two.}}
}

We refer to such data sets as \emph{non-sequential}, emphasizing the lack of ordering on their data space. 
For such sets, data compression becomes challenging.
This is because we need to identify which seemingly
different data points actually convey the same information.
This is illustrated in Figure \ref{fig:Examples}$(a)$, whereby 
in each row, two (and only two) pictures are essentially the same (equivalent) but portrayed to appear different.  
Can you tell which two? 
The first row is designed to be an easy example, however the second row requires more time, and the third row is probably too difficult by human eye.
These examples are not arbitrary, in fact they correspond to three previously studied ``non-conventional'' data models - the choice model~\cite{Farias}, the \emph{Ehrenfest} diffusion model (see~\cite{Tullio2008}, p. 5), and the graphical model (see~\cite{Choi, Kondor2008,Kondor2009}).

In this paper we extend low-dimensional \emph{linear embedding} techniques~\cite{Rudel,Achlioptas2003,Baraniuk2008,Candes2004,Blumensath2009}, to the above mentioned non-sequential data models - 
more specifically, to finite-dimensional spaces where data equivalences result from a finite group action.
We consider a two-step embedding process, illustrated in Figure \ref{fig:Examples}$(b)$.
In the first step, we utilize a special function which produces the same output if two data sets are equivalent (under this group action); such a function, termed an \emph{invariant}, accounts for data equivalence. 
Note however that the converse may not always hold, \textit{i.e.}, two data sets producing the same output may not always be equivalent, such converses are related to \emph{separating invariants}~\cite{Dufresne}, and \emph{completeness} in pattern recognition~\cite{Kakarala,Kondor2008,Kondor2009}.
In the second step, a linear embedding is applied on the output of step one, to move the data to the low-dimensional space.  
The interest here is to obtain embedding guarantees, to support the use of such techniques as a kind of compression scheme.
This has to be done with hypotheses on the discriminative power of the applied invariant, as an appropriate one-to-one embedding is not possible if the converse does not hold for any two data points of interest.

\newcommand{\LX}{\Real[\X]}
\newcommand{\G}{\mathcal{G}}
\newcommand{\X}{\mathcal{X}}
\newcommand{\V}{\mathcal{V}}

\textbf{Main results:} 
We extend two embedding theorems to finite-dimensional, non-sequential data spaces $\LX$, discussed here for the case where the group $\G$ acts by 
\emph{permutation} action.
Let $\mathcal{R}$ denote a subset of $\LX$, that contains canonical data points in $\LX$, canonical under equivalence by action of $\G$.
Then for a {bounded} set $\V$ of data points (possibly \emph{infinite}), 
assuming that the subset $\V_{\mathcal{R}}$  of canonical points ($\V$ ``projected'' onto $\mathcal{R}$),
are discriminable
by the invariant (\textit{i.e.}, satisfies the converse property), 
our extension (Theorem \ref{thm:Sauer}) of the Whitney's embedding theorem shows that  
almost all such two-step embeddings can one-to-one embed 
$\V_{\mathcal{R}}$,
if the embedding dimension exceeds $2 k$ whereby $k$ is the box-counting dimension of canonical points in set closure of $\V_{\mathcal{R}}$.
For a \emph{finite} set $\V$ of data points, 
our extension (Theorem \ref{thm:JL}) of the Johnson-Lindenstrauss lemma shows that a randomly sampled linear embedding, 
preserves isometries up to factors $1 \pm \epsilon$ with probability at least $1 - \beta$, if the embedding dimension exceeds $(2\log k + \log(1/\beta))/\alpha(\epsilon,\delta)$ for some function $\alpha$, and $\delta$ is a positive constant that upper limits a to-be-defined undiscriminable fraction, between any two canonical points in $\V_{\mathcal{R}}$.
In the second theorem, the value $k$ measuring the size of the set $\V_{\mathcal{R}}$ of canonical points, may be much smaller than that of the whole set $\V$, up to a factor $\#\G$ in group size. 
All proofs are simple and require little knowledge of invariant theory, facilitated by making obvious linear properties of invariants over a tensor space. 

\textbf{Significance of this work:} 
This is a preliminary report, on potential techniques for database compression of non-sequential data, \textit{e.g.}, DNA fragments, chemical molecular compositions, web-graph connections, record of intervallic events, etc.
Here the models to admit any type of finite group (permutation) action - more general than specific cases considered in~\cite{Lav,Choi}.
Extensions to any matrix group action seems feasible - to be pursued in future work.
A synergistic relationship is developed between linear embeddings and (data) invariants, whereby 
this work can be viewed as an adaptation of invariants for low-dimensional data in high-dimensional ambient spaces.
Provable guarantees are provided on the required \emph{storage complexity} (embedding dimension), tied directly to the size of the data set. The invariant used in the second embedding step does not determine this complexity; it only needs to satisfy the discrimability hypothesis.
While probabilistic data models are typically used in past related works~\cite{Lav,Reznik,Choi}, they are not required here.
We discuss invariants with polynomial-time \emph{computational complexity}, being at most $m n^\omega$ where $m$ is embedding dimension, $n$ is data-dimension of the model used, and $\omega \geq 1$. 
Compare with representation theoretic transform-type invariants 
(see~\cite{Kakarala,Kondor2008,Kondor2009}), where these methods require
complexity of at least $\mathcal{O}((\#\G)^2)$ to execute the fast transforms, a potentially large number if the group size $\#\G$ is huge ($\#\G$ may even be super-{exponential} in $n$ for permutation groups, see~\cite{Kondor2008}, ch. 3 \& 7).

\textbf{More discussion on related prior work:} 
Non-sequential data sets have been of interest for some years now, in pattern recognition~\cite{Kakarala}, 
probability theory~\cite{Diaconis}, machine learning~\cite{Diaconis,Huang,Kondor2008}, optimization~\cite{Absil,Chandrasekaran2010a}, choice models~\cite{Farias}, etc.
Our interest in linear embeddings is due to the wealth of recent interest on this topic, \textit{e.g.}, \emph{compressed sensing}~\cite{Candes2004}.
For invariant functions, the key area is invariant theory~\cite{Cox,Dufresne}, though there exists other guises, \textit{e.g.},
convex graphical invariants~\cite{Chandrasekaran2010a}, triple-correlation~\cite{Kakarala,Kondor2008}, see also survey article~\cite{Wood1996}.
One of their main applications of invariant theory is classification, and characterization of discriminative ability 
is of recent focus, 
see Dufresne's Ph.D thesis~\cite{Dufresne}.
For finite groups, a key result is that the set of all canonical points is in bijection with an \emph{affine algebraic variety} corresponding to the \emph{ideal of relations}, see~\cite{Cox}, pp. 345-353; however the best known complexity bound is super-exponential in the number of data-dimensions $n$.
For triple-correlation and equivalences under compact groups, Kakarala in his Ph.D thesis characterized the discriminative ability under certain conditions ~\cite{Kakarala}.
Kakarala uses representation theoretic techniques known as Tannaka-Krein duality.
The difficulty in obtaining computationally efficient invariants with absolute discriminative ability,
is appreciated by observing that even for the specific class of graphical invariants, 
a polynomial-time algorithm for \emph{graph isomorphism} is still unknown for general graphs.

The work~\cite{Lav} is mainly an information theoretic study, for an efficient algorithm specialized for multisets see~\cite{Reznik}.
In~\cite{Choi} a very efficient $\mathcal{O}(\ell^2)$ algorithm specialized for compressing $\ell$-node graphs is given, though their algorithm cannot be used as a graphical invariant.
In both~\cite{Lav,Choi}, the dimension required for appropriate compression, is similar to that of our Johnson-Lindenstrauss lemma (Theorem \ref{thm:JL}) - there will be savings logarithmic $\log(\#\G)$ in group size. 
For representation theoretic methods, partial labellings of graphical data is considered in~\cite{Kondor2009}.

\renewcommand{\fn}{\footnote{
Kakarala's formulation of homogeneous spaces is different than that of Kondor (see Supplementary Material \ref{ssect:triple}).
Kondor points out that Kakarala's definition, in some cases, ``do not model real-world problems as well''. We tend to agree.}}

For triple-correlations, Kakarala's proof in~\cite{Kakarala} is non-constructive, so an algorithm to invert an invariant function does not exist in general. 
However, invariant theory shows that the set of canonical points have a \emph{manifold}, or \emph{algebraic variety}, structure. Thus a possible future direction - inspired by compressed sensing - is to consider \emph{manifold optimization} techniques (\textit{e.g.},~\cite{Absil}) to perform inversion. 
In pattern recognition, correlation-type invariants are usually treated disparately from invariant theory, however they are related to polynomial functions from an invariant ring.
However, do note that correlation invariants restrict to only \emph{transitive} permutation group actions (where we say the data space is \emph{homogeneous}). 
Also as Kondor pointed out~\cite{Kondor2008}, pp. 89-90, one needs to take care of Kakarala's notion of  
homogeneous spaces\fn.

\newcommand{\Integers}{\mathbb{Z}}
\newcommand{\E}{\mathbb{E}}

\newcommand{\g}{g}

\renewcommand{\x}{x}

\textbf{Organization:} 
Section \ref{sect:back} touches on preliminaries, developing the type of invariants used in this work.
Section \ref{sect:Embed} states the main results, on Whitney embedding (Subsection \ref{ssect:Whit}) and Johnson-Lindenstrauss (Subsection \ref{ssect:JL}). 
Technical proofs are provided in Section \ref{sect:proofs}.

\textbf{Supplementary Material (SM-I \& SM-II):} 
For the sake of most readers who will not be familiar with both invariant theory, and representation theoretic analyses of correlation functions, two sets of supplementary materials are provided at the very end of this manuscript.
Results from both these topics, alluded to throughout this text, are summarized in these materials.

\section{Preliminaries} \label{sect:back}

\subsection{Finite-dimensional data $\G$-spaces:} \label{ssect:hom}

We assume some basic familiarity with \emph{group theory}. 
Let $\G$ denote a group, where 
$h$ and $\g$ denote group elements. 
Let $\X$ denote a set of a finite number $n$ of elements, 
and $x$ denotes an element of $\X$.
Define a \emph{permutation} action of group $\G$ on the set $\X$, where $\g(\x)$ is the image of $x$ under $g$, \textit{i.e.}, $g(x) \in \X$.
This is a \emph{left action}, \textit{i.e.}, for $h,g\in\G$ we have $(hg)(x)= h(g(x))$.
A set $\X$ endowed with such an action of $\G$ is called a \textbf{$\G$-space}.

\newcommand{\LeftAct}[2]{ {#1}^{#2} }
\newcommand{\size}[1]{\##1}
Let $\Real$ denote the set of real numbers.
Let
$\LX$ denote a set of \emph{real-valued} $n$-dimensional vectors, indexed over the set $\X$.
Data points lie in this set.
For $\a \in \LX$, the element of $\a$ indexed by $x$ is written as $a_x$ for all $x\in\X$.
The space $\LX$ (and therefore also the data) inherit the group action.
If $\LeftAct{\a}{\g}$ denotes the image of $\a$ under $\g$, \textit{i.e.}, $\LeftAct{\a}{\g} \in \LX$, then we have 
$(\LeftAct{\a}{\g})_{\g(\x)}=a_\x$ for any $\x \in \X$. 
By the left action of $\G$ on $\X$ given above, it follows that $\LeftAct{\a}{h\g} = \LeftAct{(\LeftAct{\a}{\g})}{h}$.
While $\LX$ can be identified with $\Real^n$, the 
notation $\LX$ emphasizes the group action. 
We illustrate using the following examples. 
Let $e$ denote the group identity element of $\G$, and let $\size{\X}$ be the cardinality of $\X$.

\begin{Example}\textbf{[Periodic data]:}
Let $\X = \{1,2,\cdots, n\}$.
Let $\G$ denote the $n$-th order cyclic group, \textit{i.e.}, $\G = \{e, g,g^2, \cdots, g^{n-1}\}$, 
whereby $\G$ acts on $\X$ as follows: for the special element $g$, we have $g(i) = i + 1$ for $ 1\leq i < n$, and $g(n) = 1$. 
This action is transitive. 
\end{Example}

\newcommand{\Sym}[1]{\mbox{Sym}_{#1}}
\begin{Example}\textbf{[Choice \& graphical data]:}
Let $\X$ be the set of size-$\omega$ subsets of $\{1,2,\cdots,\ell\}$, where the size $\size{\X} = {\ell \choose \omega}$.
Let $\Sym{\ell}$ be the \textbf{symmetric group} (or the group of all permutations) on $\ell$ letters.
Consider the group action of $\Sym{\ell}$ on $\X$, where for any $\g \in \Sym{\ell}$, we have the image 
$g(\mathcal{V}) = \{g(i) : i \in \mathcal{V}\}$ for any $\mathcal{V}\in \X$. 
This action is transitive. 
The special case $\omega=2$ corresponds to graphical data, as any graph is defined by the specification of ${\ell \choose 2}$ edges. 
\end{Example}

More generally, one would let $\G$ act on $\Real^n$ as a \emph{matrix group} - as in invariant theory~\cite{Cox,Dufresne}.
For simplicity, we focus only on permutation groups, which in fact covers all data models that apply for triple-correlation invariants~\cite{Kakarala,Kondor2008,Kondor2009}.

\newcommand{\LG}{\Real[\G]}
\newcommand{\at}{\mat{z}}
\newcommand{\ate}{z}
\newcommand{\A}[1]{\mathcal{A}\bI{#1}}
\newcommand{\bI}[1]{^{(#1)}}

\renewcommand{\V}{\mathcal{V}}
\newcommand{\rep}{\rho}
\newcommand{\brak}[1]{\left\langle #1\right\rangle}

\subsection{$\G$-invariants with certain linearity properties:} \label{ssect:tensor}

\newcommand{\w}{\omega}
\newcommand{\Rt}[1]{\mathcal{T}^{#1}(\LX)}
\newcommand{\tpow}{{\times \w}}
\newcommand{\opow}{{\otimes \w}}
\newcommand{\barr}{\llbracket b_\xw \rrbracket}
\newcommand{\xw}{{\mat{x}\bI{1:\w}}}
\newcommand{\xww}{{\mat{x}\bI{1:\w+1}}}
\newcommand{\Rw}{\Real[\X^\tpow]}
\newcommand{\lb}{\llbracket}
\newcommand{\rb}{\rrbracket}

We provide bare minimal background on invariant theory. 
Those familiar with this material may find our presentation unconventional, as the material is discussed in the way that we feel best supports the exposition of our main results.

We build a tensor space using the vector space $\LX$.
For $\w\geq 1$,
let $\X^\tpow$ denote the product set $\X \times \cdots \times \X$ between $\w$ copies of $\X$.
Then an \textbf{$\w$-array}, denoted $\barr$, has $n^\w$ components $b_{\xw}$ indexed over $\X^\tpow$, \textit{i.e.}, $\mat{x}\bI{1:\w} \in \X^\tpow$, 
where $\mat{x}\bI{1:\w}$ denotes the $\w$-tuple $(x\bI{1},\cdots, x\bI{\w})$.
Let $\Rw$ 
denote the set of all $\w$-arrays over $\X^\tpow$.
 
\newcommand{\e}{\mat{e}}
\renewcommand{\fn}{\footnote{If $\e_1, \e_2, \cdots, \e_n$ is a basis of $\LX$, then the $n^\w$ tensors $\e_{\sigma(1)}\otimes\e_{\sigma(2)}\otimes\cdots\otimes\e_{\sigma(\w)}$, for all $\sigma \in \Sym{\w}$, consists a basis for the tensor space $\Rw$, see~\cite{Comon}.}}
The tensor (outer) product between two elements $\a, \a'$ in $\LX$, denoted
$\a \otimes \a' $, equals $ (a_x \cdot a_y' )_{x,y \in \X}$.
Multiple tensor products, denoted $\a\bI{1} \otimes \cdots \otimes \a\bI{\w}$ for $\a\bI{j} \in \LX$, $1\leq j\leq \w$, follow similarly.
Now $\a\bI{1} \otimes \cdots \otimes \a\bI{\w} \in \Rw$, y considering the $\w$-array $\lb a_{x\bI{1}} \cdots a_{x\bI{\w}} \rb$. 
In fact $\L[\X^\tpow]$ is isomorphic to the 
space obtained by taking tensor products (between vector spaces) of $\w$ copies of $\LX$, see~\cite{Comon}.
For this reason $\Rw$ is called a \textbf{tensor space}, where 
the dimension\fn~of $\LX$ equals $n^\w$. 
For any $\a \in \LX$ , we denote $\a^\opow$  to mean $\a\otimes \cdots \otimes \a$ with $\w$ copies of $\a$.

\newcommand{\Orb}{\Omega}

We now explain how the tensor space $\Rw$ admits invariants.
Firstly, $\X^\tpow$ inherits the group action of $\G$ on $\X$, where the image $g(\xw)$ of $\xw$ under $g$ equals
$( g(x\bI{1}),\cdots, g(x\bI{\w}))$. 
This obtains an action of $\G$ on $\Rw$, where 
for any $\barr \in \Rw$, the image $g(\barr)$ under $g$ 
equals the $\w$-array $\lb b_{g^{-1}(\xw)}\rb$ (meaning that its the $\xw$-th component of the image equals $b_{g^{-1}(\xw)}$).
The previous action of $\G$ on $\X^\tpow$ induces an \emph{equivalence relation} on $\X^\tpow$, whereby $\mat{x}\bI{1:\w}_1, \mat{x}\bI{1:\w}_2 \in \X^\tpow$ are equivalent if there exists some $\g$ in $\G$ that sends $g(\mat{x}\bI{1:\w}_1) = \mat{x}\bI{1:\w}_2$, see~\cite{Tullio2008}. The equivalence classes here are called \textbf{$\G$-orbits} (on $\X^\tpow$), denoted $\Orb_{\G}(\X^\tpow)$. 
Each $\G$-orbit $\Orb_{\G}(\X^\tpow)$ will be associated with a $\w$-array $\barr$, as follows
\begin{align}
     b_\xw = 
   \begin{cases}
   		1 & \mbox{ if } \xw \in \Orb_{\G}(\X^\tpow), \\
   		0 & \mbox{ otherwise }. \\
   \end{cases} \label{eqn:orb}
\end{align}
Finally thinking of $\Rw$ as $\Real^{(n^\w)}$, define an \textbf{inner product} as
\begin{align}
	\brak{\lb a_\xw \rb, \barr} 
	= \sum_{\xw \in \X^\tpow}  a_\xw
	\cdot  b_\xw,  \label{eqn:inner}
\end{align}
and we can construct a \textbf{$\G$-invariant}, a function whose output is invariant under action of $\G$. 

\begin{proposition} \label{pro:GInv}
Let $\G$ be a finite group, with permutation action on data space $\X$.
For some $\G$-orbit $\Orb_{\G}(\X^\tpow)$ on $\X^\tpow$, where $\w \geq 1$, 
let $f_{\Orb_{\G}(\X^\tpow)} : \Rw \rightarrow \Real$ denote the mapping
\begin{align}
f_{\Orb_{\G}(\X^\tpow)} : \lb a_\xw \rb \mapsto \brak{\lb a_\xw \rb, \barr} \label{eqn:f}
\end{align}
where $\barr$ is associated with $\Orb_{\G}(\X^\tpow)$ as in \eqref{eqn:orb}.
Then $f_{\Orb_{\G}(\X^\tpow)}$ is a $\G$-invariant, \textit{i.e.}, for any $\lb a_\xw \rb\in \Rw$ we have
$f_{\Orb_{\G}(\X^\tpow)}(g(\lb a_\xw \rb)) = f_{\Orb_{\G}(\X^\tpow)}(\lb a_\xw \rb)$ for all $g\in \G$.
\end{proposition}

\begin{proof}
For brevity, write $\Orb = \Orb(\X^\tpow)$. Let $g \in \G$.
By the earlier definition of the image of $\lb a_\xw \rb$ under $g$, the value 
$f_{\Orb}(g(\lb a_\xw \rb))$ is computed by summing the coefficients $a_\xw$ supported over a subset $\V$, of the form 
$\V = \{g^{-1}(\xw) : \xw \in \Orb\}$. 
Since $\Orb$ is a $\G$-orbit, we may verify that
$\V$ is a $(g \G g^{-1})$-orbit of $\X^\tpow$, here
$g \G g^{-1}$ is a group, $g \G g^{-1} = \{g \sigma g^{-1} : \sigma \in \G\}$. 
But $g \G g^{-1}$ is an automorphism of the group $\G$, 
hence $\V=\Orb$ and we conclude the result. \qed
\end{proof}

\newcommand{\F}{\mathcal{F}}

\newcommand{\permC}{\theta_{\G,\X}}
It is important to note that the $\G$-invariant \eqref{eqn:f} is \emph{linear} in its domain $\X^\tpow$.
We extend these invariants to obtain the following linear $\G$-invariant 
$\F_\w : \Rw \rightarrow \Real^{\kappa_\w}$ of main interest, by setting 
\begin{align}
\F_\w : \lb a_\xw \rb &\mapsto (z_1,z_2, \cdots, z_{\kappa_\w}), \label{eqn:F} \\
z_i &= {\size{(\Orb_{\G,i})}}^{-\frac{1}{2}}  \cdot f_{\Orb_{\G,i}}(\lb a_\xw \rb), \nonumber
\end{align}
where $\kappa_\w$ denotes the number of different $\G$-orbits on $\X^\tpow$, numbered as $\Orb_{\G,1},\cdots, \Orb_{\G,\kappa_\w}$, 
and $\w \geq 1$.
We propose to use \eqref{eqn:F} in the first embedding step (recall illustration Figure \ref{fig:Examples}$(b)$). 
\begin{algorithm}{\textbf{$\G$-invariant \eqref{eqn:F} and embedding step one}} \label{alg:1}
\bitm \itemsep0pt
\item[1)] for given data point $\a \in \LX$, take the $\w$-th tensor power $\a^\opow$.
\item[2)] output the length-$\kappa_\w$ vector $\F_\w(\a^\opow)$.
\eitm
\end{algorithm}

In the upcoming Section \ref{sect:Embed}, the linearity of $\F_\w$ will be exploited to connect with linear embedding theory.
The normalization factor $\size{(\Orb_{\G,i})}^{-\frac{1}{2}}$ w.r.t. orbit cardinality in \eqref{eqn:F} is so that $\F_\w$ will have unity operator norm (to ensure stability).

But before going on to discussing embeddings, we clarify some properties of the invariants. 
Firstly, $\F_\w$ has \textbf{polynomial complexity} of evaluation (in $n$ for fixed $\w$), exactly $n^\w$.
Next, the number of $\G$-orbits $\kappa_\w$ over $\X^\tpow$ determines the (dimension of the) range of $\F_\w$,
and we call $\kappa_\w$ the \textbf{invariant dimension}.
We briefly discuss how to determine $\kappa_\w$.
Let $\permC : \G \rightarrow \Real$, 
that satisfies 
\begin{align}
\permC(\g) = \size{\{  x \in \X : g(x)=x \}} \label{eqn:permC}
\end{align}
for all $\g \in \G$, \textit{i.e.}, the value $\permC(\g)$ equals the number of points in $\X$ \emph{fixed} by the permutation $g$ in $\G$. 
The classical \textbf{Burnside lemma}, see \textit{e.g.}~\cite{Tullio2008}, p. 106, allows us to determine $\kappa_\w$ as follows 
\begin{align}
\kappa_\w = \frac{1}{\size{\G}} \sum_{g\in\G} \left( \theta_{\G,\X}(\g) \right)^\w. \label{eqn:permC2}
\end{align}
Note $\permC(e) = \size{\X} = n$ for the identity element $e$.

\begin{Example}\textbf{[Periodic data]:} \label{eg:ccyc2}
If $\G$ equals the cyclic group on $n$ letters, \textit{i.e.}, then 
$\permC(g) = 0$ for all  $g \neq e$. Since $\size{\G} = \size{\X} =n $, thus $\kappa_\w = n^{\w -1}$.
\end{Example}

\newcommand{\sig}{\sigma}
To simplify calculation of \eqref{eqn:permC}, 
one may use the fact that for any $g \in \G$, $\permC(\sig g \sig^{-1}) = \permC(g)$ for all $\sig \in \G$, see the following example.
There exists an equivalence relation on elements in $\G$, if we deem $h$ equivalent with $g$ if 
$h = \sig g \sig^{-1}$ for some $\sig \in \G$, see~\cite{Tullio2008}, p. 81. 

\begin{Example}\textbf{[Graphical data]:} 
For $\G = \Sym{\ell}$ with some integer $\ell$, by the above relation there exists a bijection between equivalence classes, and the unordered partitions of $\ell$, see~\cite{Tullio2008}, ch. 10.
For example, we can express $\ell = 3$ as $1+1+1$, $2+1$, and $3$; in the first partition three 1's appear, in the second partition one 1 appears and one 2 appears.
One can use this bijection to show that
$\permC(g)= \{\# \mbox{ of 2's appearing}\} + {\{\# \mbox{ of 1's appearing}\} \choose 2}$
for the partition corresponding to $g$.
\end{Example}

\begin{Remark}\label{rem:1}
In invariant theoretic terms, the $\G$-invariant $\F_\w$ is equivalent to a generating set of the degree-$\w$ homogeneous polynomials in the invariant ring, see supplementary material \ref{sm:invar1}.
Due to interest in applying invariants for classification, there is recent focus on studying minimal sets of invariants that discriminate between all data points, 
\textit{i.e.}, any $\a_1, \a_2 \in \LX$ are never mistaken if $\a_1 \neq \LeftAct{\a_2}{g}$ for all $g\in\G$, see~\cite{Dufresne} (Theorem \ref{thm:Dufresne}). 
Unfortunately such powerful discriminability properties come at super-exponential complexity (Fact \ref{fact:1}).
Thus, it is meaningful to ask, for a given invariant $\F_\w$, what are the pairs of data points that it cannot discriminate. 
For $\F_\w$, this amounts to looking at an {affine algebraic variety}, see supplementary material \ref{sm:invar2}. 
In particular for $\G$-spaces with transitive action, we can view $\F_\w$ as a multi-correlation function (see \ref{ssect:triple}), and relate to 
completeness results for the triple-correlation~\cite{Kakarala,Kondor2008,Kondor2009} (see \ref{app:Kaka}).
\end{Remark}

\section{Two Theorems on Low-Dimensional Linear Embeddings of Data-Invariants} \label{sect:Embed}
\newcommand{\att}{\a}
\subsection{Two-step linear embedding (Figure \ref{fig:Examples}$(b)$):}

\renewcommand{\fn}{\footnote{Even with normalization here, the domain of $\F_\w$ of interest is still the set of all $\w$-th tensor powers, \textit{i.e.}, $\{\a^\opow : \a \in \LX \} = \{\tilde{\a}^\opow: \a \in \LX \}$, where $\att = \a/ \Norm{\a}{2}^{\w-1}$.}}

\newcommand{\fnn}{\footnote{If $\F_\w (\a^\opow) = \F_\w ((\LeftAct{\a}{g})^\opow)$, then normalization on both sides by $c = \Norm{\a}{2}^{\w-1}= \Norm{\LeftAct{\a}{g}}{2}^{\w-1}$ preserves the equality.}} 

For some $\w\geq 1$, first apply a $\G$-invariant in Algorithm \ref{alg:1} to place the data (some $\a \in \LX$) in $\kappa_\w$ dimensions.
Next, use a linear map $\Sens : \Real^{\kappa_\w} \rightarrow \Real^m$ to effect the {dimension reduction}, whereby $m < \min(\kappa_\w,n)$. 
Specifically, compute
\begin{align}
  \Phi \left( \F_\w\left( \a^\opow\right) \right),   
  \label{eqn:hash}
\end{align}
where for convenience $ \Phi \F_\w$ will stand for the concatenation of the map $\F_\w$ followed by the map $\Sens$. 
Clearly $ \Phi \F_\w$ is a $\G$-invariant, linear in the domain $\L[\X^\tpow]$, and drops dimensions down to $m$.

We desire embeddings that map the data set, some  $\V \subset \LX$, onto the lower dimensional space in some injective manner.
This is possibly only when the embedding dimension $m$ 
is sufficiently large enough to accommodate the data set. 
The key here is that $m$ can be much smaller than 
the ambient data dimension $n$, where $m$ should really only be tied to the size of $\V$. 
Linear embeddings have been studied for when $\V$ is a union of  subspaces~\cite{Candes2004,Lu2008,Blumensath2009}, and a smooth manifold~\cite{Baraniuk2006,Clarkson2008,Absil}.
Here we look at the case where $\V$ comes from a finite-dimensional, non-sequential $\G$-spaces for finite groups $\G$.
We derive analogues of two well-known embedding theorems, in this two-step setting that employs $\G$-invariants, 
for both the Whitney embedding theorem (Subsection \ref{ssect:Whit}) and
the Johnson-Lindenstrauss lemma (Subsection \ref{ssect:JL}).

\subsection{How many dimensions are needed to embed non-sequential data?} \label{ssect:Whit}

\newcommand{\cl}[1]{\overline{#1}}
\newcommand{\bN}[1]{N_{\eps}(#1)}

In Whitney embedding we consider $\V$ to be a \emph{bounded} subset of $\LX$. 
The size of a bounded subset $\V$, will be measured by the \emph{box-counting dimension}. 
For a bounded subset $\V$, we define: i) the \emph{closure} $\cl{\V}$, and ii) the minimal number $\bN{\V}$ of boxes with sides of length $\eps$ (in $\LX$) required to cover $\V$, in a grid.
The \textbf{box-counting dimension} is then defined as 
\bea
    \boxdim(\V) =\mathop{\lim}_{\eps \rightarrow 0 } \frac{\log \bN{\V} }{-\log \eps} 
\eea
if the limit exists. 
Roughly speaking, if $\boxdim(\V)=d$, then $\bN{\V} \approx \eps^{-d}$. 
The \textbf{lower box-counting dimension}, denoted $\lboxdim(\V)$, is defined regardless by replacing the limit by $\lim \inf$.

\newcommand{\Split}[1]{\mathcal{R}}

\renewcommand{\fn}{\footnote{Since if $\Split{1}$ satisfies these conditions, then $\{\LeftAct{\a}{g} : \a \in \Split{1} \}$ for any $g \in \G$ also satisfies.}}
From our two-step embedding \eqref{eqn:hash}, the map $\Sens \F_\w$ cannot produce a 
one-to-one embedding for $\V$, since the linear tensor invariant $\F_\w$ is not always one-to-one on $\w$-th tensor powers of $\V$. 
On the other hand, we do not care to discriminate between equivalent data points.
Thus to state what is an appropriate or desirable embedding, we 
first define a canonical notion of elements in $\LX$, of which we only discriminate between. 
To this end, define the following 
disjoint subsets
of $\LX$.
For $\a \in \LX$, we say $\a$ is \emph{un-fixable} if $\LeftAct{\a}{\g} \neq \a$ is satisfied \emph{for all} $g\in \G$. 
Let $\Split{1}$ denote an open set in $\LX$.
Let $\Split{1}$ satisfy the following 3 properties: i) \emph{all} elements of $\Split{1}$ are un-fixable, ii) the $\size{\G}$ subsets $\{\LeftAct{\a}{g} : \a \in \Split{1} \}$, one for each $\g \in \G$, are \emph{disjoint}, and iii) the union $\cup_{g\in\G }\{\LeftAct{\a}{g} : \a \in \Split{1} \}$ contains \emph{all} un-fixable elements in $\LX$.
There are exactly $\size{\G}$ disjoint\fn~open sets in $\LX$ that satisfy the above properties. 
We call these open sets \textbf{fundamental regions}, and any one of them will give us our required canonical notion.
For $\V \subset \LX$, a set of canonical elements can be $\{\LeftAct{\a}{g}\in {\Split{1}} : \a \in \V, g \in \G\}$, which we denote by $\V_{\Split{1}}$ for brevity.
Our hypothesis on discriminability is now stated formally: 
a $\G$-invariant is said to be \textbf{discriminable} over a subset $\V$, if this function is one-to-one over $\V_\Split{1}$ where $\Split{1}$ is any fundamental region (note that this definition does not depend on the choice of $\Split{1}$).

The following theorem is a analogue of Theorem 2.2.~\cite{Sauer1991}, for two-step linear embeddings \eqref{eqn:hash} over finite dimensional $\G$-spaces. 

\newcommand{\Region}{\overline{\Split{1}}}
\newcommand{\Map}{\varrho}

\begin{theorem} \label{thm:Sauer}
Let $\G$ be a finite group.
Let $\X$ be a finite dimensional $\G$-space.
For some $\w \geq 1$, let $\F_\w$ be the $\G$-invariant in \eqref{eqn:F}.
Let $\Split{1}$ be any fundamental region. 

Let $\V$ denote the data set, $\V \subset \LX$, and assume $\V$ is bounded.
Assume $\F_\w$ is discriminable over $\V$, 
and let $k= \boxdim(\cl{\V_\Split{1}})$, where we assume this limit $k$ exists.

Let $\Sens$ be a linear map, that drops dimension from $\kappa_\w$ to $m$.
Then if $m > 2k$, then almost all such linear maps, the concatenated map $\Sens\F_\w$ will be 
discriminable over $\V$.
\end{theorem}

The two-step linear embedding \eqref{eqn:hash} with embedding dimension 
twice that of the data set,
is guaranteed to appropriately embed a data set $\V$ as long as the linear tensor $\G$-invariant $\F_\w$ is discriminable over $\V$.

We make three comments on Theorem \ref{thm:Sauer}, starting with storage complexity. In its original version~\cite{Sauer1991} for sequence data spaces, 
the value $k$, 
is taken as the box-counting dimension of the (closure of the) \emph{whole} data set $\V$.
We intuitively expect a ``factor of $\size{\G}$ savings'', as we only need to differentiate between canonical elements in $\Split{1}$.
Unfortunately for finite groups $\G$, the box-dimension $k= \boxdim(\cl{\V_\Split{1}})$ will always equal $\boxdim(\cl{\V})$.
However in the next subsection, we assume $\V$ to be finite, and we observe savings in Johnson-Lindenstrauss embeddings.
%
%

Secondly the computational complexity of evaluating $\Sens \F_\w$ is exactly 
$m n^\w$, polynomial in data dimension $n$ (for fixed $m,\w$). 
Each coordinate of $\Sens \F_\w$ is obtained by a weighted average of linear functions $f_{\Orb_{\G,i}(\X^\tpow)}$, $ 1\leq i \leq \kappa_\w$.

\renewcommand{\fn}{\footnote{Kondor~et.~al. represented each data corresponding to edge $\{i,j\}$, in a redundant fashion using multiple coefficients $a_x$ of $\a\in \L[\Sym{\ell}]$, for all $x$ that send $\{\ell-1,\ell\}$ to $\{i,j\}$.}}
Thirdly the linearity of $\Sens \F_\w$ may be exploited to reduce computation. 
For example in~\cite{Kondor2009}, Kondor~et.~al. used a subspace of $\L[\Sym{\ell}]$ 
to represent\fn~graphical data on $\ell$ nodes, a $(\Sym{\ell})$-space where $n=\ell!$, see~\cite{Kondor2008}.
Now if the data lives in a $k$-dimensional subspace $\V$ of $\LX$, $k < n$, let $A : \Real^k \rightarrow \V$ be a linear map onto $\V$. 
Then the tensor product map $A^\opow : \Real^{k^\w} \rightarrow \V^\opow$, where $\V^\opow \subset \L[\X^\tpow]$, is \emph{linear} in its domain $\Real^{k^\w}$. 
Now the concatenated map from $\Real^{k^\w}$ to $\Real^m$ will be $\Sens \F_\w A^\opow$, where
each coordinate is obtained by a map obtained from a weighted average of functions $f_{\Orb_{\G,i}(\X^\tpow)}A^\opow$, $1 \leq j \leq \kappa_\w$, and this map is linear (and can be evaluated in $\kappa^\w$ operations. 
Hence, the total evaluation complexity of $\Real^{k^\w}$ to $\Real^m$ equals $m k^\w$, where again $k$ is the data dimension.
In the above example where $\X = \Sym{\ell}$, we have $k = {\ell \choose 2}$, 
so the complexity equals $\mathcal{O}(m \ell^{2\w})$,
which (for fixed $m,\w$) is polynomial in the number of nodes $\ell$.

\subsection{How many dimensions are needed to preserve isometries of non-sequential data?} \label{ssect:JL}

\newcommand{\err}{\mathcal{E}_\G}

Theorem \ref{thm:Sauer} does not provide any notion of distance isometries under embedding, important for certain ``sketching''-type applications.
An important result for isometry preservation is the Johnson-Lindenstrauss lemma.
In this part, the data set $\V$ will be assumed to contain a \emph{finite} number of discrete points in $\LX$.
Also here, we state the discriminabilty hypothesis slightly differently.
By 2-norm $\Norm{\cdot}{2}$ on elements in $\L[\X^\tpow]$, we mean the norm 
\begin{align}
\Norm{\barr}{2} = \sqrt{\sum_{\xw \in \X^\tpow  }b_\xw^2}. \label{eqn:2norm}
\end{align}
as if we were treating $\L[\X^\tpow]$ as $\Real^{(n^\w)}$.
Assuming that $\F_\w$ is discriminable over $\V$, there must exist some constant $\RIC < 1$, such that 
if for any $\a_1,\a_2\in\V_\Split{1}$, 
where $\Split{1}$ is any fundamental region, 
we have
\begin{align}
 \Norm{A_{\F_\w}(\att_1^\opow- \att_2^\opow)}{2}^2 
 \leq \RIC \cdot \Norm{\att_1^\opow- \att_2^\opow}{2}^2, \label{eqn:RIC}
\end{align}
where $A_{\F_\w} : \L[\X^\tpow] \rightarrow \L[\X^\tpow]$ is the orthogonal projection onto the kernel of $\F_\w$. 
That is for canonical elements $\a_1,\a_2\in\V_\Split{1}$, the constant $\RIC$ captures the maximal fraction of ``energy'' of the error $\att_1^\opow - \att_2^\opow$ in the kernel of $\F_\w$. 

The following theorem is a analogue of (the most basic form of the) the Johnson-Lindenstrauss lemma, for two-step linear embeddings \eqref{eqn:hash} over finite $\G$-spaces.
The result is stated for the case where the coefficients of $\Sens$ are sampled from the normal distribution.
However as in many works~\cite{Candes2004,Achlioptas2003,Baraniuk2008,Rudelson2010}, extensions to more general distributions
should not be too difficult.

\begin{theorem} \label{thm:JL}
We take $\X,\G, \L[\X]$ and $\Split{1}$ as defined in Theorem \ref{thm:Sauer}.
Let $\V$ contain a finite number of discrete points in $\LX$. 
Let $k = \size{\V_\Split{1}}$. For some $\w \geq 1$, assume $\F_\w$ is discriminable over $\V$, and that the constant $\RIC < 1$ satisfies \eqref{eqn:RIC}.
Assume that the size $m \times \kappa_\w$ linear map $\Sens$, has coefficients independently sampled from a normal distribution with variance $1/m$.
Then with probability at least $1 - \beta$, if the embedding dimension $m$ of the map $\Sens$ exceeds 
\begin{align}
\frac{2\log k + \log(1/\beta)}{\alpha((\eps-\RIC)/(1-\RIC) )} \label{eqn:JLrate}
\end{align}
where $\alpha(y) = y^2 - y^3$ for any $y\in \Real$, we will have for any 
$\a_1,\a_2 \in \V$, $\a_1 \neq \a_2$, the following isometries
\begin{align}
 \Norm{\Sens \F_\w(\att_1^\opow- \att_2^\opow)}{2}^2~ 
 \begin{cases}
   ~\leq (1+\eps) \cdot \Norm{\mat{b}_1^\opow- \mat{b}_2^\opow}{2}^2, \\
   ~\geq (1-\eps) \cdot \Norm{\mat{b}_1^\opow- \mat{b}_2^\opow}{2}^2,
 \end{cases}
 \label{eqn:JL}
\end{align}
for any positive $\eps > \RIC$, 
and canonical elements $\mat{b}_1,\mat{b}_2 $ (where
$\mat{b}_1  = \LeftAct{\att_1}{g_1}$ and $\mat{b}_2  = \LeftAct{\att_2}{g_2}$ 
for some $g_1,g_2 \in \G$ such that $\mat{b}_1, \mat{b}_1 \in \V_\Split{1}$).
\end{theorem}

The factor $\eps$ in \eqref{eqn:JL} should not be too close to the constant $\RIC$ in \eqref{eqn:RIC} - this increases the required value for $m$ (it affects the denominator of \eqref{eqn:JLrate}).
As opposed to the previous Whitney embeddings, the (potential) ``factor of $\size{\G}$'' savings appear in $k$ (here $k = \size{\V_\Split{1}}$ not $k = \size{\V}$).
Do note there is a difference how these savings impact the embedding dimension $m$; unlike the previous Theorem \ref{thm:Sauer} where the factor of $\size{\G}$ impacts $m$ multiplicatively (seen from the required assumption $m > 2k$), in Theorem \ref{thm:JL} this factor impacts $m$ logarithmically (seen from \eqref{eqn:JLrate}). 
Also as seen form \eqref{eqn:JL}, 
the isometries are measured in the tensor space $\L[\X^\tpow]$ (not in 
the data space 
$\LX$). 
If one desires isometries in the original space, one requires some equivalence between the 2-norms of both spaces $\LX$ and $\L[\X^\tpow]$, not addressed here.

The next section provides technical proofs for the Theorems \ref{thm:Sauer} and \ref{thm:JL}.

\section{Technical Proofs} \label{sect:proofs}

\subsection{Proof of Theorem \ref{thm:Sauer}:}

The proof follows relatively closely with~\cite{Sauer1991}, though the consideration of $\G$-invariants allow certain simplifications, also see~\cite{Blumensath2009}.

{
\newcommand{\MapT}{A}
\newcommand{\M}{\mat{P}}
\newcommand{\alp}{\beta}
\newcommand{\Alp}{{\pmb{\beta}}}
\newcommand{\W}{\mathcal{W}}
\renewcommand{\omega}{r}
\newcommand{\Ball}[2]{\mathcal{B}_{#2}(#1)}
\renewcommand{\sig}{\sigma}

First some new notation. 
For any $\a \in \Real^n$, for some positive integer $n$, we denote $\Ball{\a,\eps}{n}$ to be the $n$-dimensional ball of radius $\eps$, centered at $\a$.
For any map, sometimes denoted $\MapT$ here, for any set $\V$ that lies in the range of $\MapT$, we shall use $\MapT^{-1}(\V)$ to denote the \emph{pre-image} of $\V$.
For any $\V\subset \Real^n$ for any $n$, we denote the volume of $\V$ as $\vol(\V)$.
We will need the following two lemmas, simplified from~\cite{Sauer1991}. 
For convenience, the lemma proofs are reproduced in Appendix \ref{app:proof}

\begin{lemma}[\textit{c.f.}~Lemma 4.2,~\cite{Sauer1991}] \label{lem:42}
For some positive integers $r,m$, $m \leq r$, let $\MapT$ be some surjective linear map from $\Real^r$ to $\Real^m$. 
Let $\sig > 0$ be a smallest singular value of $\MapT$, obtained from any matrix form for $\MapT$.
Then for any $\eps > 0$ 
\begin{align}
  \frac{  \vol(\MapT^{-1}(\Ball{\eps}{m}) \cap \Ball{\RIC}{r})}{\vol(\Ball{\RIC}{r})} 
  < 2^{r/2} \cdot \left(\frac{\eps}{\sig\RIC}\right)^m, \label{eqn:vol}
\end{align}
where $\Ball{\eps}{r}$ and $\Ball{\eps}{m}$ are respectively $r$- and $m$-dimensional balls centered at the origin.
\end{lemma}

\renewcommand{\MapT}{\rho}
\begin{lemma}[\textit{c.f.}~Lemma 4.3,~\cite{Sauer1991}] \label{lem:45}
Let $\V$ be a bounded subset of $\Real^n$, with $k = \boxdim(\overline{\V})$, and we assume this limit $k$ exists.
Let $\MapT_1,\cdots, \MapT_\omega$ be $r$ number of Lipschitz maps from $\Real^n$ to $\Real^m$. 
Further assume that for each $\a \in \V$, the linear map $A : \Real^r \rightarrow \Real^m$ described by the matrix $[\MapT_1(\a),\cdots,\MapT_\omega(\a)]$, is surjective.

For each $\Alp \in \Real^\omega$ with bounded 2-norm, $\Alp = [\alp_1,\cdots, \alp_\omega]$, define $\MapT_\Alp = \sum_{i=1}^\omega \alp_i \MapT_i$.
Then for almost every such bounded $\Alp$, the preimage $\MapT_\Alp^{-1}(\mat{0})$ of the map $\MapT_\Alp$ w.r.t. the single point $\mat{0}$, has lower box-counting dimension at most $k-m$. If $k > m$, then $\MapT_\Alp^{-1}(\mat{0})$ is empty for almost every $\Alp$. 
\end{lemma}

\renewcommand{\fn}{\footnote{Does this still work for unbounded $\Alp$?}}
}

{
\newcommand{\Alp}{{\pmb{\beta}}}
\newcommand{\MapT}{\rho}
\newcommand{\M}{\mat{D}}
\newcommand{\m}{\mat{d}}
\newcommand{\alp}{\alpha}
\renewcommand{\S}{\V\bI{2}}
\renewcommand{\e}{\mat{e}}
\begin{proof}[Proof of Theorem \ref{thm:Sauer}]

We begin by making a connection with Lemma \ref{lem:45}, first specifying for some positive integers $n_2, r$, the Lipschitz maps $\MapT_1,\cdots, \MapT_r$ (where each $\MapT_i : \Real^{n_2} \rightarrow \Real^m$), and vectors $\Alp$ in $\Real^r$.
Note, here $n_2$ replaces $n$ in Lemma \ref{lem:45}.

The domain $\Real^{n_2}$, where $n_2 = n^w$, is identified with $\L[\X^\tpow]$, 
and we set the maps $\MapT_i : \L[\X^\tpow] \rightarrow \Real^m$ as 
\begin{align}
\MapT_{i+ m(j-1)} : \lb a_\xw \rb \mapsto 
{\size{(\Orb_{\G,j})}}^{-\frac{1}{2}}  \cdot f_{\Orb_{\G,j}}(\lb a_\xw \rb) \cdot \e_i \label{eqn:thm1_1}
\end{align}
using the 1-Lipschitz functions $f_{\Orb_{\G,j}}$ appearing in \eqref{eqn:F},
for all $1 \leq i \leq m$, $1 \leq j \leq \kappa_\w$, and where $\e_1,\cdots, \e_m$ constitute any basis of $\Real^m$.
Thus here $r = m \kappa_\w$, and 
we associate each vector $\Alp$ in $\Real^{m\kappa_\w}$ with the linear map $\Sens : \Real^{\kappa_\w} \rightarrow \Real^m$, where $\Alp$ is formed by column-wise stacking of the coefficients from the matrix representation of $\Sens$. 
Under these associations, it becomes clear that the map $\MapT_\Alp : \L[\X^\tpow] \rightarrow \Real^m$ in the statement of Lemma \ref{lem:45}, equals $\Sens \F_\w$.

Let $\S = \{\a_1^\opow - \a_2^\opow : \a_1, \a_2 \in \cl{\V_\Split{1}}, \a_1 \neq \a_2\}$,
\textit{i.e.}, $\S$ is (homomorphic) to the set of non-equal pairs of $\cl{\V_\Split{1}}$.
We want to apply Lemma \ref{lem:45} with $\S$ replacing $\V$, 
with $2k$ replacing $k$ (since $\boxdim(\cl{\S}) \leq 2k$).
If the lemma applies, this shows
one-to-one mapping on $\V_\Split{1}$, which proves the theorem. 
To do so, we need to show that for each $\lb a_\xw \rb \in \V_2$, the linear map $A : \Real^{m\kappa_\w} \rightarrow \Real^m$ as described in the statement of Lemma \ref{lem:45}, is surjective. 
This will follow from the hypothesis that $\F_\w$ is discriminable over $\V$, which implies that for each $\lb a_\xw \rb \in \V_2$, there exists some function $f_{\Orb_{\G,j}}$, $1 \leq j \leq  \kappa_\w$, such that $f_{\Orb_{\G,j}}(\lb a_\xw \rb) \neq 0$. 
By the association of $A$ with the matrix $[\MapT_1(\lb a_\xw \rb),\cdots,\MapT_{m\kappa_\w}(\lb a_\xw \rb)]$, from 
\eqref{eqn:thm1_1} we conclude that since $f_{\Orb_{\G,j}}(\lb a_\xw \rb) \neq 0$ for some $j$, the map $A$ will indeed be surjective. Thus the result is proved. \qed
\end{proof}
}

The key to the proof is the discriminabilty hypothesis. 
The important point is that 
does not impact embedding dimension $m$; here $m$ is tied directly to data size (tied to $k=\boxdim(\cl{\V_\Split{1}})$).
We also point out that while Sauer et.~al. discuss more generalized versions of Lemmas \ref{lem:42} and \ref{lem:45} that do not require surjectivity of $A$ (see~\cite{Sauer1991}, Lemma 4.6), these generalizations are not useful here. This is because as our proof of Theorem \ref{thm:Sauer} reveals, the map $A$ is either surjective (in the case discriminabilty holds) or otherwise the zero-map (in the case  discriminabilty does not hold).

\subsection{Proof of Theorem \ref{thm:JL}:}

The proof here also follows with simple modifications, by appropriately incorporating discriminabilty notions. 
Standard concentration results, such as the following one, will be useful (for convenience, its proof is reproduced in Appendix \ref{app:proof}).

\newcommand{\Am}{\pmb{A}}
\newcommand{\Z}{Z}
\newcommand{\n}{\ell}
\newcommand{\h}{\theta}
\begin{lemma}[\textit{c.f.},~\cite{Baraniuk2008,Achlioptas2003}] \label{lem:conc}
Let $\Am$ be an $m \times \n$ random matrix, whose matrix entries are standard normal RVs.
Let the rows of $\Am$ be independent.
Then for any $\mat{x} \in \Real^\n$, for any $\eps > 0$ we have
\begin{align}
	\Pr\left\{ \left|~\Norm{(1/\sqrt{m})\cdot \Am\mat{x}}{2}^2  - \Norm{\mat{x}}{2}^2~\right| \leq \eps \right\}  \geq 1 - 2e^{-\frac{m}{4}(\eps^2 - \eps^3)} \label{eqn:probineq}
\end{align}
\end{lemma}

The proof of Theorem \ref{thm:JL} given below will follow for other (row independent) distributions of $\Am$, if probabilisitic inequalities similar to \eqref{eqn:probineq} are available. Indeed they are for many other of distributions, see \textit{e.g.},~\cite{Rudel,Baraniuk2008,Achlioptas2003}.
We do not go further into detail since this component is not our main focus. 
We use Lemma \ref{lem:conc} to prove our second main theorem.

\renewcommand{\fn}{\footnote{Strictly speaking, $\F_\w$ orthornormally projects onto the (coefficient space) of the complement of its kernel.}}
\begin{proof}[Proof of Theorem \ref{thm:JL}]

It suffices to show the result for pairs $\a_1,\a_2 \in \V_\Split{1}$, $\a_1 \neq \a_2$, of canonical elements, since the LHS of \eqref{eqn:JL} remains constant when replacing $\a_1,\a_2$ with $\mat{b}_1,\mat{b}_2$.
For $\Sens$ uniformly sampled (recall lemma statement) as $\Am = \Sens$, the probability that
\begin{align}
 \Norm{\Sens \F_\w(\att_1^\opow- \att_2^\opow)}{2}^2~ 
 \begin{cases}
   ~\leq (1+\eps) \cdot \Norm{\F_\w( \att_1^\opow-\att_2^\opow)}{2}^2, \\
   ~\geq (1-\eps) \cdot \Norm{\F_\w( \att_1^\opow-\att_2^\opow)}{2}^2,
 \end{cases} 
 \label{eqn:JL1}
\end{align}
holds for all ${k \choose 2} < k^2/2$  pairs whereby $\a_1,\a_2 \in \V_{\Split{1}}$, is at least $1 - k^2 \cdot e^{-\frac{m}{4}(\eps^2 - \eps^3)}$.
Here we used Lemma \ref{lem:conc} for each $\mat{x} = \Sens \F_\w(\att_1^\opow - \att_2^\opow)$, $\mat{x} \in \Real^m$. 
Comparing \eqref{eqn:JL1} with \eqref{eqn:JL}, the norm $\Norm{\cdot}{2}$ on the RHS needs to be applied on the $\L[X^\tpow]$, not $\Real^m$. 
Recall from its definition, see \eqref{eqn:F},
that $\F_\w$ is 1-Lipschitz and linear in $\Rw$, so 
the upper bound follows as
\[
\Norm{\F_\w(\att_1^\opow-\att_2^\opow)}{2}^2 \leq \Norm{(\att_1^\opow- \att_2^\opow)}{2}^2.
\]
For the lower bound, we use the hypothesis $\F_\w$ is $\RIC$-discriminable over $\V$,
where for the orthogonal projection $A_{\F_\w} : \L[\X^\tpow] \rightarrow \L[\X^\tpow]$ onto the kernel of $\F_\w$, see \eqref{eqn:RIC}, we have 
\ifdcol
\begin{align}
&\Norm{\F_\w(\att_1^\opow- \att_2^\opow)}{2}^2 + \RIC \cdot \Norm{\att_1^\opow- \att_2^\opow}{2}^2 \nn
&\geq \Norm{\F_\w(\att_1^\opow- \att_2^\opow)}{2}^2 + \Norm{A_{\F_\w}(\att_1^\opow- \att_2^\opow)}{2}^2 \nn
 &= \Norm{(\att_1^\opow- \att_2^\opow)}{2}^2, \label{eqn:orth}
\end{align}
\else
\begin{align}
\Norm{\F_\w(\att_1^\opow- \att_2^\opow)}{2}^2 + \RIC \cdot \Norm{\att_1^\opow- \att_2^\opow}{2}^2 \
&\geq \Norm{\F_\w(\att_1^\opow- \att_2^\opow)}{2}^2 + \Norm{A_{\F_\w}(\att_1^\opow- \att_2^\opow)}{2}^2 \nn
 &= \Norm{(\att_1^\opow- \att_2^\opow)}{2}^2, \label{eqn:orth}
\end{align}
\fi
equality following because both $\F_\w $ and $A_{\F_\w}$ project onto ``orthorgonal''\fn~spaces,  
which implies
\[
\Norm{\F_\w(\att_1^\opow-\att_2^\opow)}{2}^2 
 \geq (1-\RIC) \cdot \Norm{\att_1^\opow- \att_2^\opow}{2}^2.
\]
Using this in \eqref{eqn:JL1} and rearranging $(1-\eps)(1-\RIC)$, 
this proves that \eqref{eqn:JL} is satisfied with required probability, for constant $\eps(1- \RIC) +  \RIC > \eps $ (the strict inequality follows since $\RIC > 0$). 
The statement of the proposition will satisfy for some probability $\beta > k^2 \cdot e^{-\frac{m}{4}(\eps^2 - \eps^3)}$, and rescaling the $\eps$ term used here. \qed
\end{proof}

The linearity of the $\G$-invariant $\F_\w$ is very useful for deriving the lower bound \eqref{eqn:orth}), which admitted the use of orthonormality concepts.
It is also useful for deriving the upper bound, since it made it easy to check that $\F_\w$ is 1-Lipschitz.
We are now done with the proofs of both main results.

\begin{Remark}

For finite groups, there always exists an invariant satisfying the discriminability hypothesis~\cite{Dufresne} (albeit with super-exponential complexity,
see Theorem \ref{thm:Dufresne} and Fact \ref{fact:1}).
However from an embedding complexity standpoint, for any non-sequential data set, (theoretically) one can always find a two-step embedding meeting the guarantees in both Theorems \ref{thm:Sauer} and \ref{thm:JL}.

Also, the canonical points in any fundamental region $\Split{1}$, have a manifold structure within an algebraic variety (see supplementary material \ref{sm:invar2}).
Hence an interesting future direction is to connect with manifold learning techniques (\textit{e.g.},~\cite{Absil}).
\end{Remark}

\section{Conclusion} \label{sect:conc}
We present a new extension of 
linear embeddings for non-sequential data, providing two theorems
in the vein of Whitney embedding and the Johnson-Lindenstrauss lemma. 
For the latter, we show that accounting for data equivalences can provide
savings in embedding dimension up to a factor equal to the size of the invariance group (the savings is logarithmic in the second theorem). 
The extension was fairly simple, and we appeal to certain linearity properties of invariants.

\section*{Acknowledgment}
The author thanks J. Z. Sun for discussions and his reading of an initial draft, as well as R. Kakarala also for discussions and sending a copy of~\cite{Kakarala}.

\appendix 

\section{[Appendix] Proofs of Lemmas \ref{lem:42}, \ref{lem:45} and \ref{lem:conc}, appearing in Section \ref{sect:proofs}} \label{app:proof}

{
\newcommand{\MapT}{A}
\newcommand{\M}{\mat{P}}
\newcommand{\alp}{\beta}
\newcommand{\Alp}{{\pmb{\beta}}}
\newcommand{\W}{\mathcal{W}}
\renewcommand{\omega}{r}
\newcommand{\Ball}[2]{\mathcal{B}_{#2}(#1)}
\renewcommand{\sig}{\sigma}

\begin{proof}[Proof of Lemma \ref{lem:42}]
The set $\MapT^{-1}(\Ball{\eps}{m}) \cap \Ball{\RIC}{r}$ consists of points in $\Real^r$ with 2-norm at most $\RIC$, that get mapped to points in $\Real^m$ with 2-norm at most $\eps$.
Since $\MapT$ is surjective with smallest singular value $\sig > 0$, 
this set of points is contained in a cylindrical subset of $\Real^r$, with base dimension $m$, and base radius $\eps/\sigma$, see~\cite{Sauer1991}.
The volume of this cylindrical subset is at most 
$(\eps/\sig)^m \RIC^{r-m} \cdot \vol(\Ball{1}{m}) \cdot \vol(\Ball{1}{r-m})$, recall we assumed $m \leq r$.
On the other hand $\vol(\Ball{\RIC}{r}) = \RIC^r \cdot \vol(\Ball{1}{r})$.
Using these two facts 
and also the fact that the $\ell$-dimensional volumne $\vol(\Ball{1}{\ell}) = \pi^{\ell/2}/(\ell/2)!$,
we conclude \eqref{eqn:vol}. \qed 
\end{proof}

\renewcommand{\MapT}{\rho}

\begin{proof}[Proof of Lemma \ref{lem:45}]
As we consider $\Alp$ with bounded 2-norm, it suffices to replace $\Real^r$ with $\Ball{\mat{0},\delta}{\omega}$ for any $\delta > 0$, \textit{i.e.}, 
it suffices to restrict $\Norm{\Alp}{2} \leq \RIC$, for some $\RIC$ specified in the sequel.

\newcommand{\K}{{k^*}}

For any bounded $\Alp$, by assumption $\MapT_\Alp$ is Lipschitz, thus there exists some constant $C$ such that the image of any 
$\eps$-ball $\Ball{\eps}{n}$ under $\MapT_\Alp$, is contained by in some $(C\eps)$-ball in $\Real^n$ .
For $\K > 0$, consider $\eps^{-\K}$ number of $n$-dimensional $\eps$-balls, denoted $\Ball{\a_i,\eps}{n}$, with 
various centers $\a_i$ in $\V$.
If $\K > k$, we can find $\eps^{-\K}$ such balls that cover the set $\V$ of interest. 

\newcommand{\Err}{\mathcal{E}}
\renewcommand{\fn}{\footnote{For $n$ events $\Err_1,\cdots,\Err_n$, we have that the union bound $\sum_{i=1}^n\Pr\{\Err_i\} $ equals $ \sum_{i=1}^n \Pr\{\mbox{at least $i$ events $\Err_i$}\}$, see~\cite{Sathe2012}, thus we conclude that the union bound is greater than $j \cdot \Pr\{\mbox{at least $j$ events $\Err_i$}\}$ for any $j$, $1\leq j \leq n$.}}
Now for each $\Ball{\a_i,\eps}{n}$ in the covering of $\V$, the image of $\Ball{\a_i,\eps}{n}$ under $\MapT_\Alp$ contains $\mat{0}$, only if $\Norm{\MapT_\Alp(\a_i)}{2} < C \eps$ for the constants $C$ and $\eps$ above. 
For now, we make the following claim that for any $\a \in \Real^n$ and some large enough choice for $\RIC$
\begin{align}
\vol\left(\left\{\Alp \in \Ball{\RIC}{r} : \Norm{\MapT_\Alp(\a)}{2} < C \eps \right\} \right) \leq C_1 \eps^m \label{eqn:claim}
\end{align}
where $C_1$ is a positive constant. 
Then for any $\ell> 0$, by a standard argument\fn, the volume of $\Alp$ where at least $\eps^{-\ell}$ of the $\eps^{-\K}$ images of $\Ball{\a_i,\eps}{n}$ contain $\mat{0}$ (under $\MapT_\Alp$), is at most $C_1 \eps^{m-\K+\ell}$.
In other words, the preimage $\MapT_\Alp^{-1}(\mat{0})$ can be covered by less than $\eps^{-\ell}$ number of $\eps$-balls, with an exception of maps $\MapT_\Alp$ for which the volume of the corresponding $\Alp$ can be made small if $\ell > \K - m$ and $\eps$ is small. 
Thus we conclude when $\ell > \K - m$ and $\eps$ goes to $0$, we have $\lboxdim(\MapT_\Alp^{-1}(\mat{0})) \leq \ell$ for almost every $\Alp$ in $\Ball{\mat{0},\delta}{\omega}$. 
As this holds for all $\ell > \K -m$, and that $\K$ can be made arbitrarily close to $k$ for sufficiently small $\eps$, see~\cite{Sauer1991, Blumensath2009},
we have $\lboxdim(\MapT_\Alp^{-1}(\mat{0})) \leq k - m$.

We finish the proof by showing the earlier claim \eqref{eqn:claim}. 
Associate $\MapT_\Alp(\a)$ with a linear map $A$ as described in the lemma statement, whereby we assumed that $A$ is surjective.
Hence, the positive constant $\sig$ as given in the statement of Lemma \ref{lem:42} will exist.
We then can apply \eqref{eqn:vol}, by observing that the volume on the LHS of \eqref{eqn:claim}, equals 
the volume $\vol(\MapT^{-1}(\Ball{C\eps}{m}) \cap \Ball{\RIC}{r})$ similar to that the LHS of \eqref{eqn:vol} (with $\eps$ replaced by $C \eps$).
Thus for a large enough choice for $\RIC$ (where $C/(\sig \RIC) \leq 1$), we can find a constant $C_1$ that satisfies \eqref{eqn:claim}. \qed
\end{proof}

}

\begin{proof}[Proof of Lemma \ref{lem:conc}]
Express $\Norm{\Am\mat{x}}{2}^2 = \mat{x}^T(\Am^T \Am)\mat{x} = |\brak{\Am_i, \mat{x}} |^2$, where $\Am_i$ equals the $i$-th row of matrix $\Am$. 
Call $\Z_i = |\brak{\Am_i, \mat{x}} |^2$, and
observe $\E \Z_1 = \E \Z_i = \Norm{\mat{x}}{2}^2$, whereby 
without loss of generality we assume $\Norm{\mat{x}}{2}^2 = 1$.
We thus want to upper bound the probability $\Pr\{ | \sum_{i=1}^n \Z_i - m | > m \eps \}$.
We will only consider one side $\Pr\{ \sum_{i=1}^n \Z_i - m   > m \eps \}$, the other side $\Pr\{ \sum_{i=1}^n (-\Z_i) + m  > m \eps \}$ can be considered similarly.

By assumption $\Am$ has independent rows, the RV's $\Z_i$ are mutually independent. 
Then by Markov's inequality, for any $ \h> 0$
\begin{align}
\Pr\left\{ \sum_{i=1}^n \Z_i - m  > m \eps \right\}
&\leq e^{-m\h(\eps + 1) } \cdot \left( \E e^{\h \Z_1}\right)^m, \label{eqn:markov}
\end{align}
where we used the fact that $\Z_i$'s are identically distributed.
Using the fact that the entries of $\Am$ are standard normal RV's, then $\Z_1$ is \emph{chi-squared} and for $\h < 1/2$, and its a standard result that $\E e^{\h \Z_1} = (1- 2\h)^{-m/2}$. Substituting this form for $\E e^{\h \Z_1}$ in \eqref{eqn:markov}, we optimize the upper bound  over $\h$, which requires $\h = \eps/(2+2\eps) < 1/2$. 
It follows that the LHS probability of \eqref{eqn:markov} is at most
$
\left[(1+\eps) e^{-\eps}\right]^{m/2}, 
$
and what we wanted to show follows from the bound $1+ \eps \leq \exp(\eps - (\eps^2 - \eps^3)/2)$. \qed
\end{proof}

\bibliographystyle{acm}

\par\setcounter{section}{0}\setcounter{subsection}{0}\setcounter{equation}{0}\gdef\theequation{SM-\Roman{section}.\arabic{equation}}%
\gdef\thesection{SM-\Roman{section}}

\section{[Supplementary Material] Background on Invariant theory} \label{app:invar}

\newcommand{\K}{\mathbb{R}}
\newcommand{\f}{f}
\newcommand{\Ring}{\K[Z_1,\cdots, Z_n]}
\newcommand{\xo}{x_1}

\subsection{The invariant ring always satisfies the discriminability hypothesis:} \label{sm:invar1}

We expect most readers to be unfamiliar with invariant theory. 
For their convenience, this first set of supplementary material briefly covers results/facts cited and alluded to in the main text. 
We begin with the connection of invariant theory to \emph{algebraic geometry} - the study of polynomial functions/equations.
We discuss the \emph{invariant ring}, \textit{i.e.}, the ring of invariant polynomial functions. We clarify how the $\G$-invariant $\F_\w$ in \eqref{eqn:F} actually relates to such functions, hence the kernel of $\F_\w$ relates to \emph{algebraic varieties}.
We state a result on \emph{seperating invariants} from Defrusne's thesis (Theorem \ref{thm:Dufresne}), that for finite groups the invariant ring has \emph{absolute} discriminative power.
We state the results that 
how the set of canonical points has an manifold structure as an algebraic variety (Theorem \ref{thm:canonical}). 
For a good reference text see Cox-Little-O'Shea~\cite{Cox}.

We assume some basic \emph{ring theory}. 
Denote $\Ring$ to be the ring of $n$-variate polynomials over $\K$. 
For $\f \in \Ring$, let $\f$ denote an $n$-variate polynomial with real coefficients.
We think of $\f$ as a polynomial \emph{function} with domain $\Real^n$, by letting $\f(a_1, \cdots, a_n)$ be the evaluation of $\f$ at point $(a_1, \cdots, a_n) \in \Real^n$.
By the identification of $\LX$ with $\Real^n$, we 
also think of $\f$ as a function on $\LX$, for some $\G$-space $\X$ where $\size{\X} = n$. 
For some $\a \in \LX$, we write the evaluation as $\f(\a)$.

Going back to \eqref{eqn:f}, we identify $f_{\Orb_{\G}(\X^\tpow)}$ with polynomial functions in $\Ring$, as follows.
There exists some $f \in \Ring$, such that $f_{\Orb_{\G}(\X^\tpow)}(\a^\opow) = \f(\a)$ for any $\w$-th tensor powers $\a^\opow$, \textit{i.e.}, if the domain $\L[\X^\tpow]$ of the former function is restricted to tensor powers, then  the the former function is essentially a polynomial function.
This polynomial $\f$ that corresponds to $f_{\Orb_{\G}(\X^\tpow)}$ must be \emph{homogenous}, \textit{i.e.}, all monomials of $\f$ must all be of degree $\w$.

\renewcommand{\fn}{\footnote{For simplicity we still focus on permutation actions, though the invariant theoretic results discussed here holds for matrix groups in general.}}

\newcommand{\fnt}{\footnote{For matrix groups, we have a more general formula based on Molien's Theorem~\cite{Cox},~p. 340.}}

By the above association of $\G$-invariants $f_{\Orb_{\G}(\X^\tpow)}$ and polynomials $\f$, such an $\f$ is a $\G$-invariant. 
We formalize the permutation action\fn~of $\G$ on the polynomial ring $\Ring$.
Allow $\G$ to permute the variates $\Z_i$'s by the identification between $\Real^n$ and $\LX$.
More specifically for any $g \in \G$, if $\LeftAct{\f}{g}$ denotes the polynomial after permuting the variates of $\f$, then for any evaluation under $\a \in \LX$ we have $(\LeftAct{\f}{g})(\a) = \f(\LeftAct{\a}{g})$. 
Hence if the polynomial $\f$ is a $\G$-invariant, then $\f$ must satisfy 
$\LeftAct{\f}{g} = \f$ for all $g \in \G$. 
Invariant theory is the study of the set $\Ring^\G$ of all $\G$-invariant polynomials, for some group $\G$.
This set is called an \textbf{invariant ring} (of $\G$).
Now with reference to the previously discussed polynomial ring $\Ring$, 
note that $\Ring^\G$ is a subring of $\Ring$, and that $\Ring^\G$ contains the constant polynomials. 
Also $\Ring^\G$ is said to be \emph{graded}, whereby each grade refers to the set of all $\G$-invariant homogeneous polynomials of a certain degree $\w \geq 0$, see~\cite{Cox}, p. 331. 
We refer to this set of degree-$\w$ homogeneous polynomials as the \textbf{$\w$-th component} of $\Ring^\G$. 
Clearly, each $\w$-th component is closed under $\Real$-linear combinations.
In fact, it is known that each such component can be 
\emph{generated} by 
$\kappa_\w$ polynomials $\f_1, \cdots, \f_{\kappa_\w}$, each $\f_i$ corresponding to the $i$-th orbit 
invariant $f_{\Orb_{\G,i}}$, recall \eqref{eqn:F}.
It now becomes clear how the $\G$-invariant $\F_\w$ corresponds to the $\w$-th component; each ``row'' of $\F_\w$ corresponds to a (polynomial) generator. 
The number $\kappa_\w$ of generators is computable\fnt~by the same equation \eqref{eqn:permC2}. 

\renewcommand{\fn}{\footnote{The statement in~\cite{Dufresne} uses a stronger notion of discriminability, called a \emph{geometric separating set}, see Definition 3.2.1, p. 15. Also it holds for general matrix groups.}}

At this point one realizes that Algorithm \ref{alg:1} in Subsection \ref{ssect:tensor} proposes to only use one $\w$-th component. 
Evaluating $\F_\w$ only requires polynomial complexity ($n^\w$ operations).
But what about the discriminability hypothesis? 
In the next Supplementary Material \ref{app:triple}, we explain the connection between each $\F_\w$ and the so-called \emph{multi-correlations} (related to pattern recognition). 
In particular for the special case $\w=3$, Kakarala has applied representation theoretic methods to obtain so-called \emph{completeness} results, or in other words a characterization of the discriminability hypothesis
under certain conditions.
On the other hand if one is willing to consider the entire invariant ring, the discriminability hypothesis is known to \emph{unconditionally} satisfy for \emph{any} subset in $\LX$. 
We cite the following result in Dufresne's thesis, stated here slightly differently\fn.

\begin{theorem}[Corollary 3.2.12,~\cite{Dufresne},~p. 26] \label{thm:Dufresne}
Let $\G$ be a finite group. 
Let $\X$ be a finite $\G$-space.   
Then all $\w$-th components of the corresponding invariant ring, for all $\w \leq \size{\G}$, 
will be discriminable over the whole data space $\LX$.
That is for any fundamental region $\Split{1}$, for any canonical points $\a_1, \a_2 \in \LX_\Split{1}$, $\a_1,\neq \a_2$, there exists some $\G$-invariant $\f$ in $\Ring^\G$ with degree at most $\size{\G}$, such that $\f(\a_1) \neq \f(\a_2)$. 
\end{theorem}

Recall each $\F_\w$ corresponds to the $\w$-th component. 
Hence if all $\G$-invariants $\F_\w$, for all $\w \leq \size{\G}$, are appropriately made to form a single $\G$-invariant, then such a $\G$-invariant will be discriminable over any data set $\V$. 
This leads to the following important observation.

\begin{fact} \label{fact:1}
\textbf{The discriminability hypothesis can always be satisfied with large enough computational complexity:}
There exists a single $\G$-invariant corresponding to $\w$-components, $\w \leq \size{\G}$, that for any data set $\V \subset \LX$, satisfies the discriminablity hypothesis in both our Whitney embedding Theorem \ref{thm:Sauer} and Johnson-Lindenstrauss Theorem \ref{thm:JL}. 

This implies that any bounded, non-sequential data set $\V$ can be appropriately embedded with embedding dimension $m$ tied only to its relevant size $k$. 

However, such an invariant requires $\mathcal{O}(n^{\size{\G}})$ complexity to compute, exponential in the size of $\G$ - clearly infeasible in practice for most group sizes.
\end{fact}

It is not yet known if the size requirements on $\w$ in Theorem \ref{thm:Dufresne} is necessary (in certain cases they can be improved). 
Now since the same theorem holds for all of $\LX$, one meaningful approach would be relax this requirement, and only consider \emph{specific} subsets of $\LX$.
Kakarala adopts a similar strategy for triple-correlations, by obtaining completeness results under certain assumed data conditions (see second set of supplementary material). 

\subsection{The set of canonical points includes a manifold structure:} \label{sm:invar2}

Another beautiful aspect of invariant theory, is due to its connection with \emph{algebraic geometry}.
In particular, there is a remarkable explanation how
the set of all canonical points has a manifold-like structure, in the form of an \emph{affine algebraic variety}~\cite{Cox}, pp. 345-353.
An (affine) algebraic variety is a set of points, whereby there exists a set of polynomial equations, for which is satisfied by every point in this set.
For example, the kernel of the $\G$-invariant $\F_\w$ in \eqref{eqn:F} is related to the following algebraic variety
\bea
\{\a \in \LX : \f_i(\a) = 0, 1\leq i \leq \kappa_\w\}, \label{eqn:kernel}
\eea
where $\f_i(\a) = f_{\Orb_{\G,i}(\X^\tpow)}(\a^\opow) $. 
For the same polynomials $\f_i$, the following set is also an algebraic variety
\bea
\{(\a_1,\a_2) \in \LX \times \LX : \f_i(\a_1) - \f_i(\a_2) = 0, 1\leq i \leq \kappa_\w\}, \label{eqn:kernel2}
\eea
whereby this second set \eqref{eqn:kernel2} contains pairs of points in $\LX$ that \emph{cannot} be discriminated by the $\G$-invariant $\F_\w$.
In theory, the set could be computed by \emph{elimination theory} using a \emph{Gr\"obner basis}, see~\cite{Cox}, ch. 3, which will obtaining useful characterizations of such pairs of points $(\a_1,\a_2)$. 
Though such an approach can be unwieldy for large $n$, it does suggest a a possible algebraic geometry view of characterizing discrimability of invariants, besides the representation theoretic techniques of Kakarala's.
Also Kakarala's techniques currently only hold for triple-correlations (\textit{i.e.}, $\w=3$), whereas here $\w$ could be arbitrary. 

\newcommand{\MapT}{\rho}
\newcommand{\Comp}{\mathbb{C}}
\newcommand{\CRing}{\Comp[Z_1,\cdots, Z_n]}
\newcommand{\CLX}{\Comp[\X]}
\newcommand{\p}{\beta}
\newcommand{\RingY}{\Comp[Y_1,\cdots, Y_\ell]}

The algebraic variety structure of the set of canonical points
is a little more complicated to explain, and requires the \emph{algebraic closure} of $\Real$ to the complex field $\Comp$.
Take a generating set of the invariant ring $\CRing^\G$ over $\Comp$, say $\f_1,\cdots, \f_\ell$ for some $\ell \geq 1$, and form a map $\MapT : \CLX \rightarrow \Comp^\ell : \a \mapsto (\f_1(\a), \cdots, \f_\ell(\a))$, where $\CLX$ is the \emph{complexification} of $\LX$.
Recall the notation $\CLX_\Split{1}$, which means a set a canonical points in $\CLX$ lying in some fundamental region $\Split{1}$.
There exists an invariant theoretic result that says that $\CLX_\Split{1}$ is in bijection with the \emph{image} of $\MapT$, 
whereby this image is actually an algebraic variety.
The set of polynomial equations that describe the image comes from the generators of a special \emph{ideal} of the ring $\RingY$ of $\ell$-variate polynomials, where $\ell$ is the number of generators $\f_i$ of the invariant ring. 
This ideal, known as the \textbf{ideal of relations}, contain all $\p$ in $\RingY$ whereby $\p(\f_1,\cdots,\f_\ell)$ is identically zero; here $\p(\f_1,\cdots,\f_\ell)$ is thought of as a polynomial in the variates $Z_i$'s.
This result is stated as follows.

\begin{theorem}[Theorem 10,~\cite{Cox}, p. 351] \label{thm:canonical}
Let $\f_1,\cdots, \f_\ell$ generate the invariant ring $\CRing^\G$, for some $\ell \geq 1$.
Let $\MapT : \CLX \rightarrow \Comp^\ell : \a \mapsto (\f_1(\a), \cdots, \f_\ell(\a))$.

Let $\p_1,\cdots, \p_r$ generate the ideal of relations in the ring $\RingY$, for some $r \geq 1$. Consider the algebraic variety 
\bea
\{(b_1,\cdots, b_r) \in \Comp^r : \p_i(b_1,\cdots, b_r)=0, 1 \leq i \leq r\} \label{eqn:image}
\eea
Then the image of $\MapT$ is surjective over the algebraic variety \eqref{eqn:image}. 
In fact if we restrict $\MapT$ over the domain $\CLX_\Split{1}$ for any fundamental region $\Split{1}$, then $\MapT$ with this restriction of domain, becomes bijective.
\end{theorem}

Theorem \ref{thm:canonical} remarkably shows how the set of canonical points, after passing through this map $\MapT$, has the manifold structure of the algebraic variety \eqref{eqn:image}.
This brings to mind the possibility of applying \emph{manifold learning} techniques to learn the canonical points. However until one derives an analogue of Theorem \ref{thm:canonical} for the reals, one needs to work in $\Comp$.

\section{[Supplementary Material] Completeness results for triple-correlation} \label{app:triple}

\newcommand{\gw}{{\mat{g}\bI{1:\w}}}
\newcommand{\gmw}{{\mat{g}\bI{1:\w-1}}}

\renewcommand{\t}{t}

\subsection{Multi-correlations are connected with invariant theory:} \label{ssect:triple}

Auto- and triple-correlation functions have been employed as invariants in pattern recognition~\cite{Kakarala,Kondor2008,Kondor2009}, though the presentation has always been disparate from invariant theory.
The first goal of this second set of supplementary material, is to provide unification.
We begin by clarifying how a generalization of such functions (that we call \emph{multi-correlations}) are one and the same to the graded components of the invariant ring (see previous Supplementary Material \ref{app:invar}).
Then next, for the sake of most readers not familiar with Kakarala's completeness results for the triple-correlation, we provide a primer in Subsection \ref{app:Kaka}).

For correlation functions studied pattern recognition, the group action is limited to \emph{transitive} permutation actions.
Recall the two examples given in Subsection \ref{ssect:hom}. 
For this special case, the $\G$-space $\X$ is referred to as a \textbf{homogeneous space}. 
To explain correlations, we require the following notion of $\G$ itself as a homogeneous space. 

\begin{Example}\textbf{[$\G$ as a homogeneous space]:} \label{eg:3}
For an abstract group $\G$, define a action of $\G$ on itself, 
where for any $\g \in \G$, we have the image $g(\sigma) = g \sigma$ for any $\sigma \in \G$, 
\textit{i.e.}, $\G$ acts on itself by  left multiplication.
This is a transitive action, so $\G$ (as a set) is a homogeneous space.
\end{Example}

The last example admits discussion of the vector space $\LG$; we consider $\G$ as the set $\X$. 
Let $\at$ denote an element in $\LG$, where $\ate_g$ denotes an indexed element of $\at$ for $g \in \G$.
For any $\at \in \LG$, the \textbf{multi-correlation} $\A{\w}_{\at}$ for some $\w \geq 1$, is given as
\begin{align}
\A{\w}_{\at}(\g_1,\cdots, \g_{\w-1}) &= \sum_{\sigma \in \G} \ate_{\sigma } \ate_{\sigma g_1}\cdots  \ate_{\sigma g_{\w-1}}, \label{eqn:MC}
\end{align}
where for $j$, $1\leq j < \w$ we have $g_j \in \G$. 
The cases $\w=2$ and $\w=3$ specialize respectively to the auto- and triple-correlations.
For any $\w \geq 1$, the function $\A{\w}_{\at}$ is a $\G$-invariant, \textit{i.e.}, for any $\alpha \in \G$, we have $\A{\w}_{\LeftAct{\at}{\alpha}} =\A{\w}_{\at} $; to verify this, simply evaluate \eqref{eqn:MC} with $\LeftAct{\at}{\alpha}$ and put
$(\LeftAct{\at}{\alpha})_\sigma = \ate_{\alpha^{-1}\sigma}$ for any $\sigma \in \G$.

While the (correlation) functions \eqref{eqn:MC} seem to be only defined for the space $\LG$, we can accommodate any $\G$-space $\X$, by \emph{extending} elements in $\LX$ to $\LG$.
Let $\xo$ denote an element in $\X$ that has been (arbitrarily) chosen and fixed.
Using this $\xo$ then for any $\a \in\LX$, the \textbf{extension} of $\a$, denoted $\bar{\mat{a}}$, satisfies
\begin{align}
   \bar{a}_g = a_{g(x_1)}, ~~\mbox{ for all } g \in \G. \label{eqn:bar_a}
\end{align}
The \textbf{stabilizer} of the fixed element $\xo$,
denoted $\S_{\xo}$, is the set of group elements in $\G$ that leave $\xo$ un-moved, \textit{i.e.}, $\S_{\xo} = \{g \in \G : g(\xo) = \xo \}$. Clearly $\S_{\xo}$ will be a subgroup of $\G$.
Since we do not discuss other stabilizer subgroups in the sequel, we will drop the subscript $\xo$ from $\S_{\xo}$ and simply write $\S$ throughout.
The relationship \eqref{eqn:bar_a} relates $\S$ to extensions of vectors in $\LX$, whereby
note that any extension $\bar{\mat{a}}$ 
is \emph{constant over left-cosets} of $\S$ in $\G$, \textit{i.e.}, for any $g \in \G$, we have $\bar{a}_{gs} = \bar{a}_g$ for any $s \in \S$.
Hence when considering homogeneous spaces $\X$
we only need to evaluate \eqref{eqn:MC} (for $\A{\w}_{\bar{\a}}$ where $\a \in \LX$) at points 
$\{(t_{i_1},\cdots, t_{i_{\w-1}}) : 1 \leq i_1,\cdots, i_{\w-1} \leq n \}$, 
where each $\t_j$ is a \textbf{left-coset representative}. 
There are at most $n^{\w-1}$ such points, where $n = \size{\X}$.
For the previously fixed $\xo$, enumerate the rest of the elements in $\X$ as $x_2,x_3,\cdots, x_n$, and fix $t_j$ to send $ x_1 $ to $x_j$ (possible only when $\G$ acts transitively on $\X$). 
Note $n = \size{\X} = \size{\G}/\size{\S}$.
To conclude, 
extensions allow us to synonymously discuss correlations for $\LG$, and $\LX$ for any homogeneous $\G$-space $\X$.

\newcommand{\xmw}{\mat{x}\bI{1:\w-1}}
\newcommand{\tpoW}[1]{{\times (#1)}}

We proceed to show how the multi-correlation \eqref{eqn:MC} for some $\w \geq 1$, is related to the $\w$-th component of the invariant ring. 
We do this by specifying the connection with $\G$-invariant $\F_\w$ in \eqref{eqn:F}, which was already established to ``generate'' the $\w$-th degree polynomials in the ring.
For any $\a \in \LX$, we calculate the multi-correlation $\A{\w}_{\bar{\a}}$ as follows 
\ifdcol
\begin{align}
\A{\w}_{\bar{\mat{a}}}&(t_{i_1},\cdots, t_{i_{\w-1}})  \nn
&= \sum_{\sigma \in \G} \bar{a}_\sigma \bar{a}_{\sigma t_{i_1}} \cdots \bar{a}_{\sigma t_{i_{\w-1}}} \nn
&\stackrel{(a)}{=} \sum_{j=1}^n \sum_{s \in \S} \bar{a}_{t_j s} \bar{a}_{t_j s t_{i_1}} \cdots \bar{a}_{t_j s t_{i_{\w-1}}} \nn
&\stackrel{(b)}{=} \sum_{j=1}^n \sum_{s \in \S} a_{x_j} a_{(t_j s t_{i_1})(x_1)} \cdots a_{(t_j s t_{i_{\w-1}})(x_1)} \nn 
&= \sum_{j=1}^n a_{x_j} \sum_{s \in \S}  a_{(t_j s)(x_{i_1})} \cdots a_{(t_j s)(x_{i_{\w-1}})}
\label{eqn:A2simp}
\end{align}
\else
\begin{align}
\A{\w}_{\bar{\mat{a}}}(t_{i_1},\cdots, t_{i_{\w-1}})  
&= \sum_{\sigma \in \G} \bar{a}_\sigma \bar{a}_{\sigma t_{i_1}} \cdots \bar{a}_{\sigma t_{i_{\w-1}}} \nn
&\stackrel{(a)}{=} \sum_{j=1}^n \sum_{s \in \S} \bar{a}_{t_j s} \bar{a}_{t_j s t_{i_1}} \cdots \bar{a}_{t_j s t_{i_{\w-1}}} \nn
&\stackrel{(b)}{=} \sum_{j=1}^n \sum_{s \in \S} a_{x_j} a_{(t_j s t_{i_1})(x_1)} \cdots a_{(t_j s t_{i_{\w-1}})(x_1)} \nn 
&= \sum_{j=1}^n a_{x_j} \sum_{s \in \S}  a_{(t_j s)(x_{i_1})} \cdots a_{(t_j s)(x_{i_{\w-1}})}
\label{eqn:A2simp}
\end{align}
\fi
where in $(a)$ we apply $\sigma = t_j s$ for some $t_j$, in $(b)$ we apply \eqref{eqn:bar_a} and $\bar{a}_{t_j s} = a_{(t_js)(x_1)} =a_{t_j(x_1)} =a_{x_j}$, and the last equality follows by definition $t_j(x_1) = x_j$.
We notice the following from the final expression \eqref{eqn:A2simp}. 
For each $j$, $1\leq j \leq n$, the second summation really runs over indexes over $\X^\tpoW{\w-1}$ in the set
$\{t_j (\xmw) : \xmw \in \Orb_{\S}(\X^\tpoW{\w-1}) \}$, where $\Orb_{\S}(\X^\tpoW{\w-1})$ is the $\S$-orbit (over $\X^\tpoW{\w-1}$) that contains $(x_{i_1},\cdots, x_{i_{\w-1}})$.
The LHS and RHS of \eqref{eqn:A2simp} are really determined by the indices $i_1, \cdots, i_{\w-1}$, for at most $n^{\w-1}$ such choices.

We notice the following connection between the final expression in \eqref{eqn:A2simp} and the $\G$-invariant as applied in Algorithm \ref{alg:1}.
First, there is a one-to-one correspondence between $\G$-orbits on $\X^\tpow$, and $\S$-orbits on $\X^\tpoW{\w-1}$.
This correspondence is obtained for $\Orb_{\G}(\X^\tpow)$, by identifying $\Orb_{\S}(\X^\tpoW{\w-1})$ with the subset $ \{\xmw : (\xmw, x_1) \in \Orb \}$ of $\X^\tpoW{\w-1}$.
Secondly 
for any $\G$-orbit $\Orb = \Orb_{\G}(\X^\tpow)$ on $\X^\tpow$, by the corresponding $\w$-array $\barr$ in \eqref{eqn:orb}, 
we can express (see \eqref{eqn:F})
\begin{align}
f_{\Orb}(\a^\opow) 
 &= \sum_{j =1}^n a_{x_j} \left( \sum_{(x\bI{1},\cdots, x\bI{\w-1}) \in \Orb'_j} a_{x\bI{1}}\cdots a_{x\bI{\w-1}} \right)
 \nonumber
\end{align}
where for each $j$, $1\leq j \leq n$, we have $\Orb'_j = \{\xmw : (\xmw, x_j) \in \Orb \}$.
Note that $\Orb_j'$ is simply an orbit of the subgroup $t_j \S t_j^{-1}$ that stabilizes $x_j$, whereby $\Orb'_j = \Orb_{\S}(\X^\tpoW{\w-1})$, the $\S$-orbit previously identified with the $\G$-orbit $\Orb$.
Recall from the proof of Proposition \ref{pro:GInv} that the ($t_j \S t_j^{-1}$)-orbit is simply the set 
$\{t_j(\xmw) : \xmw \in \X^\tpoW{\w-1} \}$. 
Finally, compare with \eqref{eqn:A2simp} by taking $(x_{i_1},\cdots, x_{i_{\w-1}}) \in \Orb_{\S}(\X^\tpoW{\w-1})$ (determined by the indices $i_1, \cdots, i_{\w-1}$), and conclude the following result.

\begin{proposition} \label{pro:multiEquiv}
Let $\a \in \LX$. 
Let $\Orb_{\S,1}(\X^\tpoW{\w-1}),\cdots, \Orb_{\S,\kappa_\w}(\X^\tpoW{\w-1})$ denote the $\kappa_\w$ number of $\S$-orbits on $\X^\tpoW{\w-1}$.
Then firstly for an extension $\bar{\a}$, the multi-correlation $\A{\w}_{\bar{\a}}$ has at most $\kappa_\w$ unique evaluations, found at the points $(t_{i_1},\cdots, t_{i_{\w-1}})$ corresponding to the representatives $(x_{i_1},\cdots, x_{i_{\w-1}})$ of the $\S$-orbits.

Secondly, the output $\F_\w(\a^\opow)$ of Algorithm \ref{alg:1} is 
equivalent to the multi-correlation $\A{\w}_{\bar{\a}}$ for the extension $\bar{\a}$, whereby evaluation at the point $(t_{i_1},\cdots, t_{i_{\w-1}})$ corresponding to $(x_{i_1},\cdots, x_{i_{\w-1}})$, is equal to the value of $f_{\Orb_{\G}(\X^\tpow)}(\a^\opow)$, see \eqref{eqn:F}, where the $\G$-orbit $\Orb_{\G}(\X^\tpow)$ corresponds to the $\S$-orbit that contains $(x_{i_1},\cdots, x_{i_{\w-1}})$.
\end{proposition}

The second part of Proposition \ref{pro:multiEquiv} proves the intended equivalence between the  $\G$-invariants in \ref{eqn:F} and the multi-correlations. 
This proposition establishes a connection between Kakarala's representation theoretic analysis, discussed in the sequel, and the invariant theory discussed in Supplementary Material \ref{app:invar}.

\subsection{Kakarala's completeness results for triple-correlation:} \label{app:Kaka}
This subsection provides a brief introduction to \emph{representation theoretic} techniques for showing completeness of the triple correlation. 
We discuss a constructive algorithm for finite cyclic groups (which more generally also applies to finite abelian groups), and Kakarala's completeness result for compact groups.
Note that compact groups include finite groups under the discrete topology.
Good references to this material include the textbook~\cite{Tullio2008}, and Kakarala's and Kondor's theses~\cite{Kakarala,Kondor2008}.

Here we let $\V$ denote a finite-dimensional vector space. 
A \textbf{representation} of a group $\G$ over $\V$, is an action of $\G$ on the vector space $\V$; 
for any $\g \in \G$, each $\at \in \V$ is sent to $\rep(\g)\at$, whereby any $\rep(\g)$ is an invertible linear map.
For example suppose $\V = \LG$, and for $g \in \G$ set $\rep(\g)$ to be a 0-1 matrix in $\Real^{\G \times \G}$
whose $h,\sigma$-th element $(\rep(\g))_{h,\sigma}$ equals 1 i.f.f. $ h = g \sigma$.
This representation, called the \textbf{left-regular representation}, is in fact related to the previous example of $\G$ acting on itself (\textit{i.e.}, $\G$ is a homogeneous $\G$-space).

\newcommand{\dual}{\widehat{\G}}

A representation $(\rep,\V)$ is said to be \textbf{irreducible}, if the subspace of $\V$ invariant under the representation action, is trivial (\textit{i.e.}, the invariant subspace equals either $\{\mat{0}\}$ or $\V$).
An \textbf{unitary} representation $(\rep,\V)$ preserves the inner product on $\V$, \textit{i.e.}, for all $g\in\G$ we have $\brak{\rep(g)\at, \rep(g) \at'} = \brak{\at,\at'}$ for any $\at,\at'\in \V$.
Two representations $(\rep_1,\V_1)$ and $(\rep_2,\V_2)$ are said to be \textbf{equivalent}, if there exists a linear bijection $A : \V_2 \rightarrow \V_1 $ such that $\rep_1(g) A = A \rep_2$ for all $g\in\G$.
The \textbf{dual} of a finite group $\G$, denoted $\dual$, is the complete set of irreducible pairwise non-equivalent (unitary)
representations of $\G$. 
If $\G$ is finite then so is $\dual$.
The machinery to obtain $\dual$, from the left-regular representation, is given by the \emph{Peter-Weyl theorem} (see~\cite{Tullio2008}, pp. 85-86, for the statement for finite $\G$).
The following is the analogue of the Fourier transform, stated for finite $\G$.

\begin{Definition}[\textit{c.f.}, \cite{Tullio2008}, p. 99] \label{defn:FT}
Let $\at \in \LG$. Let $\G$ be a finite group with finite dual $\dual$.
The (abstract) \textbf{Fourier transform} component of $\at$ with respect to a irreducible (unitary) representation $(\rep,\V)$, is the linear operator $\hat{\at}(\rep) : \V \rightarrow \V$ defined by 
\bea
\hat{\at}(\rep) = \sum_{g\in \G} z_g \cdot \rep(\g). \label{eqn:FT}
\eea
\end{Definition}

\renewcommand{\fn}{\footnote{The $\mathcal{B}$ stands for \emph{bi-spectrum}, a term for the (2-dimensional) Fourier transform of the triple correlation.}}
\newcommand{\ddual}{\widehat{\G \times \G}}
\newcommand{\B}{\mathcal{B}_{\at}}
\newcommand{\mult}{m}
The techniques here will be very related to this Fourier transform.
In what follows, we need to consider the product group $\G \times \G$, and its dual $\ddual$.
Here,
each $(\rep,\V) \in \ddual$ has maps $\rep(g,h)$ indexed by an element pair $g,h \in \G$.
For the the triple correlation $\A{3}_{\at}$ of any $\at \in \LG$, we now elucidate an illuminating structure of a Fourier transform component, specially\fn~denoted $\B(\rho)$.
Consider two elements $\at_1,\at_2 \in \L[\G \times \G]$ related to $\at \in \LG$, as follows.
For $\at_1$, set $(\at_1)_{(g,g)} = z_g$ for all $g\in \G$ and $(\at_1)_{(g,h)} = 0$ when $h \neq g$.
For $\at_2$, set $(\at_2)_{(g,h)} = z_g z_h$ for all $g,h\in \G$.
Let $\dagger$ denote \emph{complex conjugation}.
Then for any $(\rep,\V) \in \ddual$, we see that
\ifdcol
\begin{align}
 & \left( \widehat{\at_1}(\rho) \right)^\dagger \widehat{\at_2}(\rho) \label{eqn:FT_AC} \\
 &=  \left( \sum_{\sigma \in \G} z_\sigma \cdot \rep(\sigma^{-1},\sigma^{-1}) \right) \cdot \left( \sum_{g,h \in \G} z_g z_h \cdot \rep(g,h) \right)  \nn
 &=  \sum_{h,g \in \G} \sum_{\sigma \in \G} z_\sigma z_g z_h \cdot \rep(\sigma^{-1}g,\sigma^{-1}h) \nn
 &=  \sum_{h,g \in \G} \sum_{\sigma \in \G} z_\sigma z_{\sigma g}  z_{ \sigma h} \cdot \rep(g,h) \nn
  &=\sum_{g,h\in \G} \A{3}_{\at} (g,h) \cdot \rep(g,h) = \B(\rho),  \nonumber
\end{align}
\else
\begin{align}
  \left( \widehat{\at_1}(\rho) \right)^\dagger \widehat{\at_2}(\rho) 
 &=  \left( \sum_{\sigma \in \G} z_\sigma \cdot \rep(\sigma^{-1},\sigma^{-1}) \right) \cdot \left( \sum_{g,h \in \G} z_g z_h \cdot \rep(g,h) \right)  \nn
 &=  \sum_{h,g \in \G} \sum_{\sigma \in \G} z_\sigma z_g z_h \cdot \rep(\sigma^{-1}g,\sigma^{-1}h) \nn
 &=  \sum_{h,g \in \G} \sum_{\sigma \in \G} z_\sigma z_{\sigma g}  z_{ \sigma h} \cdot \rep(g,h) \nn
  &=\sum_{g,h\in \G} \A{3}_{\at} (g,h) \cdot \rep(g,h) = \B(\rho),  \label{eqn:FT_AC}
\end{align}
\fi
where the second last equality follows from the definition \eqref{eqn:MC} of the triple correlation $\A{3}_{\at}$. 

\renewcommand{\fn}{\footnote{The direct sum $\V_1 \oplus \V_2$ of vector spaces equals 
$\{\mat{v}_1 + \mat{v}_2 : \mat{v}_1 \in \V_1, \mat{v}_2 \in \V_2\}$).}}

We proceed to further manipulate the LHS of \eqref{eqn:FT_AC}.
Each $(\rep, \V)$ in $\ddual$ can be expressed as 
$(\rep_1 \otimes \rep_2, \V_1 \otimes \V_2)$, where $(\rep_1,\V_1), (\rep_2,\V_2)\in\dual$, where $\rho(g,h) = \rep_1(g) \otimes \rep_2(h)$, see~\cite{Tullio2008}, p. 272.
Thus for $\widehat{\at_2}(\rho)$ in \eqref{eqn:FT_AC}, $\rho = \rho_1 \otimes \rho_2$, we conclude
\bea
\widehat{\at_2}(\rep_1\otimes\rep_2) &=& \hat{\at}(\rep_1) \otimes \hat{\at}(\rep_2), \label{eqn:z2}
\eea
where the RHS are two Fourier transforms of $\at$ in $\LX$, corresponding to representations 
$(\rep_1,\V_1), (\rep_2,\V_2)\in\dual$.
Next we require the notion\fn~of a \emph{direct sum representation} $(\rep_1 \oplus \rep_2,\V_1 \oplus \V_2)$ of two representations $(\rep_1, \V_1) $ and $(\rep_2, \V_2) $ of $\G$,
where $\V_1,\V_2$ are orthogonal.
In the direct sum for all $\g \in \G$., we mean that $\varrho_1(g)$ leaves $\V_2$ invariant, and $\varrho_2(g)$ leaves $\V_1$ invariant.
The tensor product representation $\rep_1 \otimes \rep_2$ can be expressed as direct sums of representations in $\dual$, \textit{i.e.}, 
\bea
\rep_1 \otimes \rep_2 
\equiv \bigoplus_{\varrho \in \dual} \varrho^{\otimes \mult_{\rep_1,\rep_2}(\varrho) } \label{eqn:TP_decomp}
\eea
where $\equiv$ denotes equivalence in representations (under some linear operator $A_{\rep_1,\rep_2} : \V \rightarrow \V'$ where $\V'$ is some subspace of $\LX$), 
and the notation $\varrho^{\otimes \ell }$ for $\varrho \in \dual$, $\ell \in \Integers$, means the representation $\varrho \otimes \cdots \otimes \varrho$ formed by $\ell$ copies of $\varrho$,
and finally $\mult_{\rep_1,\rep_2} : \dual \rightarrow \Integers$ returns for each $\varrho$ in $\dual$, the number of copies in the tensor product. 
From \eqref{eqn:TP_decomp} we can conclude for $\widehat{\at_1}(\rho)$ in \eqref{eqn:FT_AC}, where $\rho = \rho_1 \otimes \rho_2$,
\bea
\widehat{\at_1}(\rep_1\otimes\rep_2)  \label{eqn:z1}
&\equiv& \bigoplus_{\varrho \in \dual} \left( \hat{\at}(\varrho) \right)^{\otimes \mult_{\rep_1,\rep_2}(\varrho) }
\eea
where $\equiv$ means the same equivalence earlier in \eqref{eqn:TP_decomp}.
By the identity $\B(\rep) = \left( \widehat{\at_1}(\rho) \right)^\dagger \widehat{\at_2}(\rho)$
developed in \eqref{eqn:FT_AC}, we conclude where $\rep = \rep_1 \otimes \rep_2$ the following
\ifdcol
\begin{align}
&\B(\rep_1\otimes\rep_2) A_{\rep_1\otimes \rep_2} \label{eqn:Kaka} \\
 &=
\left[\bigoplus_{\varrho \in \dual} \left( \hat{\at}(\varrho) \right)^{\otimes \mult_{\rep_1,\rep_2}(\varrho) } \right]^\dagger
A_{\rep_1\otimes \rep_2}
\hat{\at}(\rep_1)\otimes \hat{\at}(\rep_2) \nonumber.
\end{align}
\else
\begin{align}
\B(\rep_1\otimes\rep_2) A_{\rep_1\otimes \rep_2} 
 &=
\left[\bigoplus_{\varrho \in \dual} \left( \hat{\at}(\varrho) \right)^{\otimes \mult_{\rep_1,\rep_2}(\varrho) } \right]^\dagger
A_{\rep_1\otimes \rep_2}
\hat{\at}(\rep_1)\otimes \hat{\at}(\rep_2). \label{eqn:Kaka} 
\end{align}
\fi
where $A_{\rep_1\otimes \rep_2}$ makes the equivalence \eqref{eqn:TP_decomp}.
From \eqref{eqn:Kaka}, we can now describe an algorithm that recovers the Fourier coefficients $\hat{\at}(\rho)$ from that of the triple-correlation $\A{3}_{\at}$ (\textit{i.e.}, from $\B(\rep_1\otimes\rep_2)$).
Then by a \emph{Fourier inversion theorem},~\cite{Tullio2008}, p. 100, we contain obtain from $\hat{\at}(\rho)$ the data $\at$.

\newcommand{\1}{\mathbb{1}}
\renewcommand{\fn}{\footnote{This steps of this algorithm was not stated as clearly in previous work, hence it is valuable to record them here.}}

A condition will be required for the algorithm to work: 
\begin{align}
\mbox{for all}~(\rep,\V) \in \dual,~~~\hat{\at}(\rep) \mbox{ is an \emph{invertible} map}. \label{eqn:cond_algo}
\end{align}
If \eqref{eqn:cond_algo} holds, then for all $\rep_1,\rep_2 \in \dual$ the following quantity 
\begin{align}
&\B'(\rep_1\otimes\rep_2)  = \B(\rep_1\otimes\rep_2) A_{\rep_1\otimes \rep_2} \hat{\at}^{-1}(\rep_1)\otimes \hat{\at}^{-1}(\rep_2) A_{\rep_1\otimes \rep_2}^\dagger
 \label{eqn:kaka2}
\end{align}
is well-defined, where $A_{\rep_1\otimes \rep_2}^\dagger$ is the adjoint of $A_{\rep_1\otimes \rep_2}$ with complex conjugation.
Let $(\1,\V)$ denote the 
\textbf{trivial representation}
whereby $\1(g) = 1$ for all $g\in\G$. 
We see that
\bea
 \hat{\at_1}(\1 \otimes \rep) &= & \hat{\at}(\rep), \label{eqn:triv1}\\
 \hat{\at_2}(\1 \otimes \rep) &= & \hat{\at}(\1) \cdot \hat{\at}(\rep), \nonumber
\eea
which follows from  \eqref{eqn:z1} and \eqref{eqn:z2}.
Then from \eqref{eqn:Kaka} the following algorithm\fn, under the existence of an appropriate labeling $\varrho_1,\varrho_2,\varrho_3, \cdots$ of representations in $\dual$ (where $\varrho_1 = \1$), will perform the promised task.

\begin{algorithm}{\textbf{To obtain Fourier coefficients $\hat{\at}(\rho)$ from $\B(\rho_1\otimes \rho_2)$, where $\rho,\rho_1,\rho_2\in\dual$}} \label{alg:2}

\bitm 
\item As $\B(\1\otimes\1) = \left( \hat{\at}(\1) \right)^3$ holds from \eqref{eqn:FT_AC} and \eqref{eqn:triv1}, compute $\hat{\at}(\1)= \hat{\at}(\varrho_1)$.
\item As $\B(\1\otimes\varrho_2) = \hat{\at}(\1) \cdot \left(\widehat{\at}(\varrho_2)\right)^\dagger \hat{\at}(\varrho_2)$ holds from \eqref{eqn:FT_AC} and \eqref{eqn:triv1}, compute $\hat{\at}(\varrho_2)$. 
\bitm
\item Note that since $\hat{\LeftAct{\at}{\alpha}}(\varrho_2) = \rho_2(\alpha) \hat{\at}(\varrho_2)$ for any $\alpha \in \G$, we can only determine $\hat{\at}(\varrho_2)$ up to $\G$-invariance (\textit{i.e.}, if $\hat{\at}(\varrho_2)$ solves the above expression, then so does $\rho_2(\alpha) \hat{\at}(\varrho_2)$ for any  $\alpha \in \G$).
\eitm
\item For $\varrho_3, \varrho_4, \cdots$, use the following iteration derived from both \eqref{eqn:Kaka} and \eqref{eqn:kaka2}. For $\ell \geq 3$, use
\[
\B'(\varrho_{\ell-1}\otimes\varrho_2) = \hat{\at}(\varrho_\ell)^\dagger \oplus  M_{\ell-1}
\]
to solve for $\hat{\at}(\varrho_\ell)$ where the LHS will be known using previous computations. 
where $M_{\ell-1}$ is the remainder term in the RHS of \eqref{eqn:TP_decomp} for $\varrho_{\ell-1}\otimes\varrho_2$, after pulling out one copy of $\varrho_\ell$. 
\eitm
\end{algorithm}

Now for the final step of Algorithm \ref{alg:2} to work, 
the labeling $\varrho_1,\varrho_2,\varrho_3, \cdots$ must allow $\hat{\at}(\varrho_\ell)$ to be pulled out in each $\ell$-th step. Unfortunately in general for finite groups $\G$, this labeling is unknown. On the other hand if $\G$ is cyclic, the representations $(\varrho_\ell,\V)\in \dual$, $1 \leq \ell \leq \size{\G}$, possess a ``cyclic group structure'', see~\cite{Tullio2008}, p. 274. In particular, there exists some choice for labeling $\varrho_1,\varrho_2,\varrho_3, \cdots$, such that we can express for any $2 \leq \ell \leq \size{\G}$
\[
\varrho_{\ell} = \varrho_{\ell-1} \otimes  \varrho_2 
\]
using some special choice for $\varrho_2$.
Hence for finite cyclic groups, Algorithm \ref{alg:2} will work as long condition \eqref{eqn:cond_algo} is met.
Also for finite abelian groups in general, which are always isomorphic to direct product of a finite number of finite cyclic groups, appropriate extensions can be perused.
In conclusion, Algorithm \ref{alg:2} is a constructive proof of a completeness result (under the above appropriate conditions), that $\A{3}_{\at} = \A{3}_{\at'} $ if and only if $\at'$ must be some obtainable from $\at$ by some $g\in\G$.

Using the condition \eqref{eqn:cond_algo}, Kakarala proved a remarkable completeness result of the same vein, for the large class of compact groups (which also includes some infinite groups - under appropriate generalization of the vector space $\LG$, the Fourier transform in Definition \ref{defn:FT}, and the dual $\dual$, see~\cite{Kakarala} for details).

\begin{theorem}[\textit{c.f.},~\cite{Kakarala}] \label{thm:Kaka}
Let $\G$ be a compact group, and let $\dual$ be its dual.
Let $\at$ be any arbitrary function in $\LG$, for which we assume that condition \eqref{eqn:cond_algo} is met.
Then the triple-correlation $\A{3}_{\at}$ of $\at$, equals another $\A{3}_{\at'}$ for some $\at' \in \LG$, if and only if there 
$\at' = \LeftAct{\at}{g}$ for some $g$ in $\G$.
\end{theorem}

Unfortunately Kakarala's proof is non-constructive, and we still do not know how to run Algorithm \ref{alg:2} for general groups (but see~\cite{Kakarala} for an algorithm that works for the group of all $2 \times 2$ unitary matrices with determinant $+1$). 
The proof of Theorem \ref{thm:Kaka} relies on Tannaka-Krein duality (Proposition 1,~\cite{Chevalley}, p. 199).

Note the following important points.
Note Theorem \ref{thm:Kaka} only requires condition \eqref{eqn:cond_algo} (\textit{i.e.}, does not require the labeling $\varrho_1,\varrho_2,\varrho_3, \cdots$), whereby one seems to be able to satisfy it by slight perturbation of $\a$.
This is mis-leading, as Kondor pointed out~\cite{Kondor2008}, pp. 89-90, for extensions as in  \eqref{eqn:bar_a}, \textit{i.e.}, for $\at = \bar{\a}$ for $\a \in \LX$ of general homogeneous $\G$-spaces $\X$, the condition \eqref{eqn:cond_algo} turns out be mostly unsatisfied. 
While Kakarala has yet another remarkable completeness result for homogeneous spaces (see~\cite{Kakarala}, Theorems 4.6 \& 4.7), however as Kondor also pointed out (p. 91), this result applies only for elements in $\LG$ that are constant under \emph{right cosets} of $\S$ (or invariant under left $\S$-translation as in~\cite{Kakarala}), as opposed to our definition \eqref{eqn:bar_a} which makes extensions constant over \emph{left cosets} of $\S$.
Hence Kakarala's result does not apply exactly to our setup.

In conclusion, there exists some powerful results (\textit{e.g.}, Algorithm \ref{alg:2} and Theorem \ref{thm:Kaka}) developed for the triple-correlation.
However for general groups, there is room to improve these results, especially worthwhile would be a completeness result for homogeneous spaces for extensions as defined in \eqref{eqn:bar_a}.

\end{document}